\documentclass[noinfoline,stslayout]{imsart}

\usepackage{amsfonts,amsmath,amssymb,theorem,graphicx,hyperref,color,natbib}
\usepackage{mathrsfs} 
\usepackage{multirow}
\usepackage{color}
\usepackage{caption}
\usepackage{subcaption}
\usepackage{latexsym}
\usepackage{nicefrac}
\usepackage{verbatim}
\newcommand{\R}{\mathbb R}
\newcommand{\dd}{\mathrm{d}}

\newcommand{\model}{\mathfrak{M}}

\newtheorem{exemp}{Example}[section]
\newtheorem{proposition}{Proposition}
\newtheorem{theorem}{Theorem}
\newenvironment{example}
{\begin{exemp}\begin{em}}
{\end{em}\end{exemp}}

\newcommand\bx{\mathbf{x}}

\newcommand\MF{{\mathfrak{M}}}

\begin{document}

\begin{frontmatter}

\title{Bayesian hypothesis testing as a mixture estimation model\protect\thanksref{T1}}
\runtitle{Bayesian hypothesis testing via mixture models}
\thankstext{T1}{Kaniav Kamary, 
CEREMADE, Universit{\' e} Paris-Dauphine, 75775 Paris cedex 16, France, {\sf kamary@ceremade.dauphine.fr}, 
Kerrie Mengersen, QUT, Brisbane, QLD, Australia, {\sf k.mengersen@qut.edu.au}, 
Christian P. Robert and Judith Rousseau, CEREMADE, Universit{\' e} Paris-Dauphine, 
75775 Paris cedex 16, France {\sf xian,rousseau@ceremade.dauphine.fr}. Research partly supported by the Agence Nationale de la Recherche (ANR,
212, rue de Bercy 75012 Paris) through the 2012--2015 grant ANR-11-BS01-0010 ``Calibration'' and by two 2010--2015 and
2016--2021 senior chair grants of Institut Universitaire de France. Thanks to Phil O'Neall and Theo Kypriaos from the University of Nottingham for
their hospitality and a fruitful discussion that drove this research towards new pathways. We are also grateful to the
participants of BAYSM'14, ISBA 16, several seminars, and to the members of the Bristol Bayesian Cake Club for their comments.}

\begin{aug}
 \author{\snm{{\sc Kaniav Kamary}}}
 \affiliation{Universit{\'e} Paris-Dauphine, CEREMADE}
 \author{\snm{{\sc Kerrie Mengersen}}}
 \affiliation{Queensland University of Technology, Brisbane}
 \author{\snm{{\sc Christian P.~Robert}}}
 \affiliation{Universit{\'e} Paris-Dauphine, CEREMADE,~Dept. of Statistics, University of Warwick, and CREST, Saclay}
 \author{\snm{{\sc Judith Rousseau}}}
 \affiliation{University of Oxford, Dept.~of Statistics and Universit{\'e} Paris-Dauphine, CEREMADE}
\end{aug}

\begin{abstract} We consider a novel paradigm for Bayesian testing of hypotheses and Bayesian model comparison. 
Instead of the traditional comparison of posterior probabilities of the competing hypotheses, given the data, we consider 
the hypotheses as components of a mixture model. We therefore replace the original testing
problem with an estimation one that focus on the probability or weight of a given hypothesis
within a mixture model as the parameter of interest and the posterior distribution of this weight as the
outcome of the test. 
A major differentiating feature of this approach is that that generic improper priors are acceptable. 
For example, a reference Beta $\mathcal{B}(a_0,a_0)$ prior on the mixture weight parameter can be used for
the common problem of testing two contrasting hypotheses. In this case, the sensitivity of the posterior estimates of the
weights to the choice of $a_0$ vanishes as the sample size increases, leading to a consistent procedure and a suggested 
default choice of $a_0=0.5$. Another feature of this easily implemented alternative to the
classical Bayesian solution is that the speeds of convergence of the posterior mean of the weight and of the
corresponding posterior probability are quite similar.  \end{abstract}


\end{frontmatter} 
 
\section{Introduction}\label{intro}

\subsection{A open problem}

Statistical testing of hypotheses and the related model choice problem are central issues for statistical inference.
Perspectives and methods relating to these issues have been developed over the past two centuries, but it remains an
object of study and debate, in particular because the most standard approach based on $p$-values is open to misue and 
abuse, as highlighted by a recent ASA warning \citep{wasserstein:lazar:2016}, and also because the classical (frequentist) and Bayesian paradigms
\citep{neyman:pearson:1933b,jeffreys:1939,berger:sellke:1987,casella:berger:1987,gigerenzer:1991,berger:2003,
mayo:cox:2006,gelman:2008} are at odds, both conceptually and practically. 

From the perspectives of Neyman-Pearson and Fisher, tests are constructed as a competition beween so-called null and
alternative hypotheses, and typically evaluated with respect to their ability to control the type I error, i.e., the probability 
of falsely rejecting the null hypothesis in favour of the alternative. These procedures therefore handle the
two hypotheses differently, with an imbalance that makes the subsequent action of accepting the null hypothesis problematic.
The ASA statement \citep{wasserstein:lazar:2016} recommends against basing decisions solely on the $p$-value and advocates
the use of supplementary indicators such as those obtained by Bayesian methods.

However, from a Bayesian perspective, the handling of hypothesis testing is also
problematic, for different reasons \citep{jeffreys:1939,bernardo:1980,berger:1985,aitkin:1991,
berger:jefferys:1992,desantis:spezzaferri:1997,bayarri:garcia:2007,christensen:johnson:wesley:branscum:2011,johnson:rossel:2010,
gelman:carlin:stern:etc:2013,robert:2014}. In particular, we consider that the issue of non-informative Bayesian hypothesis testing is
still mostly unresolved both theoretically and in practice, despite having produced much debate and proposals;
witness the specific case of the Lindley or Jeffreys--Lindley paradox
\citep{lindley:1957,shafer:1982,degroot:1982,robert:1993b,lad:2003,spanos:2013,sprenger:2013,robert:2014} and 
discussions about pseudo-Bayes factors
\citep{aitkin:1991,berger:pericchi:1996,ohagan:1995,ohagan:1997,aitkin:2010,gelman:robert:rousseau:2013}. 

There are similar difficulties with Bayesian model selection. Several perspectives can be defended even
for the canonical problem of comparing models with the aim of choosing one of them. For instance, \cite{hoeting:etal:1999} propose a model
averaging approach; \cite{christensen:johnson:wesley:branscum:2011} argue that
this is a decision issue that pertains to testing;  \cite{robert:2001} express this as a model index estimation
setting; \cite{gelman:carlin:stern:etc:2013} prefer to rely on more exploratory predictive tools, and \citep{pc08,rockova:george:2015}
restrict the inference to a maximisation problem in a sparse model or variable selection context.

\subsection{The traditional Bayesian framework}

Under all the frameworks described above, a generally accepted perspective is that hypothesis testing and model selection do not primarily 
seek to identify which alternative or model is ``true" (if any). From a Bayesian perspective, hypotheses can be formulated as models and 
hypothesis testing can therefore be viewed as a form of model selection, in which the aim to compare several potential statistical models 
in order to identify the model that is most strongly supported by the data 
\cite[see, e.g.][]{berger:jefferys:1992,madigan:raftery:1994,balsubramanian:1997,mackay:2002,
consonni:forster:larocca:2013}. For this reason, we will use the terms hypotheses and models interchangeably in the context of 
hypothesis testing and model choice.

The most common approaches to Bayesian hypothesis testing in practice are posterior probabilities of the model given the data \citep[see, e.g.,][]{robert:2001},
the Bayes factor \citep{jeffreys:1939} and its approximations such as the Bayesian information criterion (BIC)
and the Deviance information criterion (DIC) \citep{schwarz:1978,csiszar:shields:2000,
spiegelhalter:best:carlin:linde:2002,forbes:2006,plummer:2008} and posterior predictive tools and their variants 
\citep{gelman:carlin:stern:etc:2013,vehtari:ojanen:2012}.
For example, consider two families of models, one for each of the hypotheses under comparison,
$$
\MF_1:\ x\sim f_1(x|\theta_1)\,,\ \theta_1\in\Theta_1 \quad\text{and}\quad
\MF_2:\ x\sim f_2(x|\theta_2)\,,\ \theta_2\in\Theta_2\,.
$$
Following \cite{berger:1985} and \cite{robert:2001}, a standard Bayesian approach is to 
associate with each of those models a prior distribution, 
$$
\theta_1\sim\pi_1(\theta_1) \quad\text{and}\quad \theta_2\sim\pi_2(\theta_2)\,,
$$
and to compute the marginal or integrated likelihoods
$$
m_1(x) = \int_{\Theta_1} f_1(x|\theta_1)\,\pi_1(\theta_1)\,\dd\theta_1 \quad\text{and}\quad m_2(x) = \int_{\Theta_2}
f_2(x|\theta_2)\,\pi_1(\theta_2) \,\dd\theta_2
$$
either through the Bayes factor or through the posterior probability, respectively: 
$$
\mathfrak{B}_{12} = \frac{m_1(x)}{m_2(x)}
, \quad \mathbb{P}(\MF_1|x) = \frac{\omega_1 m_1(x)}{\omega_1 m_1(x)+\omega_2 m_2(x)}, \quad \omega_1+\omega_2=1, \quad \omega_i\geq 0
$$
Note that the latter quantity depends on the prior weights $\omega_i$ of both models. 
The Bayesian decision step proceeds by comparing the Bayes factor $\mathfrak{B}_{12}$ to the
threshold value of one or comparing the posterior probability $\mathbb{P}(\MF_1|x)$ to a bound derived from a 0--1 loss
function or a ``golden" bound like $\alpha=0.05$ inspired from frequentist practice
\citep{berger:sellke:1987,berger:boukai:wang:1997, berger:boukai:wang:1999,berger:2003,ziliak:mccloskey:2008}. As a
general rule, when comparing more than two models, the model with the largest posterior probability is selected, but
this rule is highly dependent on the prior modelling, even with large datasets, which makes it hard to promote as the
default solution in practical studies.

\subsection{Our alternative proposal}

Given that the difficulties associated with the traditional handling of posterior probabilities for Bayesian testing
and model selection are well documented and comprehensively reviewed \citep{vehtari:lampinen:2002} and \cite{vehtari:ojanen:2012},
we do not replicate or dwell further on this discussion, nor do we attempt to resolve the attendant problems. 
Instead, we propose a different approach, which we argue provides a convergent and naturally interpretable solution, a
measure of uncertainty on the outcome, a wider range of prior modelling, and straightforward calibration tools. 
 
The proposed approach is described here in the context of a hypothesis test or
model selection problem with two alternatives.  We represent the distribution
of each individual observation as a two-component mixture between both models
$\MF_1$ and $\MF_2$. The resulting mixture model is thus an {\em encompassing
model}, as it contains both models under comparison as special cases. While
those special cases are extreme cases which weights are located at the
boundaries of the interval $(0,1)$, the posterior distribution of this weight parameter does 
concentrate on one of those boundaries for enough data from the corresponding
model. This concentration property follows from the results of \cite{rousseau:mengersen:2011},
in which they established that over-fitted mixtures (i.e., mixtures with more
components than are supported by the data) can be consistently estimated,
despite the parameter sitting on one (or several) boundary(ies) of the
parameter space. The outcome of our analysis is a posterior distribution 
over the parameters of the mixture, including the component weights. 

While mixtures are natural tools for classification and clustering, which can separate a sample into observations associated with each model, 
we argue that the posterior distribution on the weights provides a relevant indicator of the strength of support for each model given the data. 
Given a sufficiently large sample from model $\MF_1$, say, this distribution will almost surely concentrate near $0$. Starting
from this theoretical garantee, we can calibrate the degree of concentration near $0$ versus $1$ for the current sample
size by comparing the posterior distribution to posteriors associated with simulated samples from models $\MF_1$ and
$\MF_2$, respectively. Even though we fundamentally object to returning a
decision about ``the" model corresponding to the data \citep{mcshane:etal:2018}, and
thus would like to halt the statistician's input at returning the above posterior, it is furthermore straightforward to
produce a posterior median estimate of the component weight that can be compared with realisations from each model.
Quite obviously, the mixture posterior produced by a standard Bayesian analysis \citep{fruhwirth:2006} also provides
information about the model parameters and the presence of potential outliers in the data.

With regard to the classical approach to Bayesian hypothesis testing, this
mixture representation is not equivalent to the use of a posterior probability.
In fact, a posterior estimate of the mixture weight cannot be viewed as a proxy
to the numerical value of this posterior probability, which we do not see as a
worthwhile tool for testing for reasons given below.  As mentioned in the
previous paragraph, this new tool can be calibrated in its own right, while
allowing for a degree of uncertainty in the hypothesis evaluation, which is not
the case for the Bayes factor.  In particular, while posterior
probabilities are scaled against the $(0,1)$ interval, it can be argued
\citep{fraser:2011,fraser:wu:sun:2009,fraser:etal:2016} that they cannot be
taken at face value because of their lack of frequentist coverage and hence
need to be calibrated as well.  Furthermore, the mixture approach offers the
valuable feature of limiting the number of parameters in the model and hence is
in keeping with Occam's razor, see, e.g., \cite{jefferys:berger:1992,rasgha2001a}.

The plan of the paper is as follows. Section \ref{sec:zero} provides a description of the mixture model specifically
created for this setting and presents a simple example of its implementation.
Section \ref{sec:consix} expands \cite{rousseau:mengersen:2011} to provide conditions on the hyperameters of the mixture model 
that are sufficient to achieve convergence.
The performance of the mixture approach is then illustrated through three further examples in Section \ref{sec:firsT} and 
concluding remarks about the general applicability of the method are made in Section \ref{sec:quatr}.

\vspace{0.3cm}
\section{Testing problems as estimating mixture models}\label{sec:zero}

\subsection{A new paradigm for testing}\label{sec:newP}

Following from the above, given two classes of statistical models,
$\MF_1$ and $\MF_2$, 
which may correspond to a hypothesis to be tested and its alternative, respectively,
it is always possible to embed both models within an encompassing mixture model
\begin{equation}\label{eq:mix}
\MF_\alpha:\ x\sim \alpha f_1(x|\theta_1) + (1-\alpha) f_2(x|\theta_2)\,,\ 0\le \alpha\le 1\,.
\end{equation}
Indeed, both models correspond to very special cases of the mixture model, one for $\alpha=1$
and the other for $\alpha=0$ (with a slight notational inconsistency in the indices).\footnote{The choice of possible
encompassing models is obviously unlimited: for instance, a Geometric mixture (Meng, 2016, personnal communication) 
$$
x\sim f_\alpha(x) \propto f_1(x|\theta_1)^\alpha\,f_2(x|\theta_2)^{1-\alpha}
$$
is a conceivable alternative. However, such alternatives are less practical to manage, starting with the issue of the
intractable normalizing constant. Note also that when $f_1$ and $f_2$ are Gaussian densities, the Geometric mixture
remains Gaussian for all values of $\alpha$. Similar drawbacks can be found with harmonic mixtures.}

When considering a sample $(x_1,\ldots,x_n)$ from one of the two models, the mixture representation still holds at the
likelihood level, namely the likelihood for each model is a special case of the weighted sum of both likelihoods.
However, this is not directly appealing for estimation purposes since it corresponds to a mixture with {\em a single
observation}.  See however \cite{o'neill:kypraios:2014} for a computational solution based upon this representation. 

What we propose in this paper is to draw inference on the individual mixture representation \eqref{eq:mix}, acting as if
each observation was individually and independently\footnote{Dependent observations like Markov chains can be modeled 
by a straightforward extension of \eqref{eq:mix} where both terms in the mixture are conditional on the relevant past observations.}
produced by the mixture model. Hence $\alpha$ represents the probability that a new observations is sampled from $f_1 $ 
belonging to model $\MF_1$. The approach proposed here therefore aims at answering the question: what is the proportion of the data which support one model, which has a definite predictive flavour to the testing problem. 

Here are five advantages we see about this approach.

First, if the data were indeed generated from model $\MF_1$, then the Bayesian estimate of the weight $\alpha$ and the posterior probability 
of model $\MF_1$ produce equally convergent indicators of preference for this model (see Section \ref{sec:consix}).
Moreover, the posterior distribution of $\alpha$ evaluates more thoroughly the
strength of the support for a given model than the single figure outcome of a Bayes factor or of a posterior
probability, while the variability of the posterior distribution on $\alpha$ allows for a more thorough assessment of the
strength of the support of one model against the other. Indeed, the approach allows for the possibility that, for a finite dataset,
one model, both models or neither model could be acceptable, as illustrated in Section \ref{sec:firsT}. 

Second, the mixture approach also removes the need for artificial prior probabilities on model indices, $\omega_1$ and $\omega_2$.
These priors are rarely discussed in a classical Bayesian approach, even though they linearly impact on the posterior
probabilities. Under the new approach, prior modelling only involves selecting an operational prior on $\alpha$, for intance a Beta 
$\mathcal{B}(a_0,a_0)$ distribution, with a wide range of acceptable values for the
hyperparameter $a_0$, as demonstrated in Section \ref{sec:consix}. While the value of $a_0$ impacts the posterior
distribution of $\alpha$, it can be argued that (a) it nonetheless leads to an accumulation of the mass near $1$ or $0$; 
(b) a sensitivity analysis on the impact of $a_0$ is straightforward to carry out; and (c) in most settings the approach 
can be easily calibrated by a parametric bootstrap experiment, so the prior predictive error can be directly estimated and
can drive the choice of $a_0$ if need be. 

Third, the problematic computation \citep{chen:shao:ibrahim:2000,marin:robert:2010} of the marginal likelihoods is bypassed, 
since standard algorithms are available for Bayesian mixture estimation \citep{richardson:green:1997,berkhof:mechelen:gelman:2003,
fruhwirth:2006,lee:marin:mengersen:robert:2008}. Moroever, the (simultaneously conceptual {\em and} computational) difficulty of 
``label switching" \citep{celeux:hurn:robert:2000, stephens:2000b,jasra:holmes:stephens:2005} that plagues both Bayesian 
estimation and Bayesian computation for most mixture models completely vanishes in this particular context, since components are 
no longer exchangeable in the current framework. In particular, we compute neither a Bayes factor\footnote{Using a Bayes factor to test for the number of 
components in the mixture \eqref{eq:mix} as in \cite{richardson:green:1997} would be possible. However, the outcome would fail to
answer the original question of selecting between both (or more) models.} nor a posterior probability related with the
substitute mixture model and we hence avoid the difficulty of recovering the modes of the posterior
distribution \citep{berkhof:mechelen:gelman:2003, lee:marin:mengersen:robert:2008,rodriguez:walker:2014}. Our perspective
is completely centred on estimating the parameters of a mixture model where both components are always identifiable.

Fourth, the extension to a finite collection of models to be compared is straightforward, as this simply involves a larger
number of components. The mixture approach allows consideration of all these models at once rather than engaging in 
costly pairwise comparisons. It is thus possible to eliminate the least likely models from simulations, since they will not be explored by the corresponding computational algorithm \citep{carlin:chib:1995,richardson:green:1997}.

Finally, while standard (proper and informative) prior modeling can be
painlessly reproduced in this novel setting, non-informative (improper) priors
are also permitted, provided both models under comparison are first
reparameterised so that they share parameters with common meaning. For
instance, in the special case when all parameters make sense in both
models,\footnote{While this may sound like an extremely restrictive requirement
in a traditional mixture model, we stress here that the presence of common meaning
parameters becomes quite natural within a testing setting. To wit, when
comparing two different models for the {\em same} data, moments like
$\mathbb{E}[X^\gamma]$ are defined in terms of the observed data and hence {\em
should} be the {\em same} for both models.  Reparametrising the models in terms
of those common meaning moments does lead to a mixture model with some and
maybe {\em all} common parameters.} the mixture model \eqref{eq:mix} can read
as
$$
\MF_\alpha:\ x\sim \alpha f_1(x|\theta) + (1-\alpha) f_2(x|\theta)\,,0\le \alpha\le 1\,.
$$
For instance, if $\theta$ is a location parameter, a flat prior $\pi(\theta)\propto 1$ can be used with
no foundational difficulty, in contrast to the traditional testing case \citep{degroot:1973,berger:pericchi:varshavsky:1998}.

\subsection{A Normal Example}\label{sec:eg1}
In order to illustrate the proposed approach, consider a hypothesis test between a Normal $\mathcal{N}(\theta_1,1)$ 
and a Normal $\mathcal{N}(\theta_2,2)$ distribution.
We construct the mixture so that the same location parameter $\theta$ is used in both the Normal $\mathcal{N}(\theta,1)$ 
and the Normal $\mathcal{N}(\theta,2)$ distribution. This allows the use of Jeffreys' (1939) noninformative prior $\pi(\theta)=1$, 
in contrast with the corresponding Bayes factor. We thus embed the test in the mixture of Normal models, 
$\alpha \mathcal{N}(\theta,1)+(1-\alpha)\mathcal{N}(\theta,2)$,  and adopt a Beta $\mathcal{B}(a_0,a_0)$ prior on $\alpha$.
In this case, considering the posterior distribution on $(\alpha, \theta)$, conditional on the allocation vector $\zeta$, leads to
conditional independence between $\theta$ and $\alpha$: 
$$
\theta|{\mathbf x,\zeta} \sim \mathcal{N}\left(\frac{n_1\bar x_1 +.5n_2\bar
x_2}{n_1+.5n_2},\frac{1}{n_1+.5n_2}\right)\,,
\quad
\alpha|\mathbf{\zeta} \sim \mathcal{B}e(a_0+n_1,a_0+n_2)\,,
$$ 
where $n_i$ and $\bar x_i$ denote the number of observations and the empirical mean of the observations allocated to
component $i$, respectively (with the convention that $n_i\bar{x}_i=0$ when $n_i=0$. Since this conditional posterior
distribution is well-defined for every possible value of $\zeta$ and since the distribution $\zeta$ has a finite
support, $\pi(\theta|x)$ is proper.

Note that for this example, the conditional evidence $\pi(x|\zeta)$ can easily be
derived in closed form, which means that a random walk on the allocation space $\{1,2\}^n$ could be implemented. We did
not follow that direction, as it seemed unlikely such a random walk would have been more efficient than a
Metropolis--Hastings algorithm on the parameter space only.

In order to evaluate the convergence of the estimates of the mixture weights, we simulated $100$
$\mathcal{N}(0, 1)$ datasets. Figure \ref{normb2} displays the range of the posterior means and medians of $\alpha$ when
either $a_0$ or $n$ varies, showing the concentration effect (with a lingering impact of $a_o$) when $n$ increases.
We also included the posterior probability of $\MF_1$ in the comparison, derived from the Bayes factor
$$
\mathfrak{B}_{12} ={2^{\nicefrac{n-1}{2}}}\big/{\exp\nicefrac{1}{4}\sum_{i=1}^n(x_i-\bar{x})^2}\,,
$$
with equal prior weights, even though it is not formally well defined since it is based on an improper prior. The shrinkage of
the posterior expectations towards $0.5$ motivate the use the posterior median instead of the posterior mean as the relevant
estimator of $\alpha$. The same concentration phenomenon occurs for the $\mathcal{N}(0, 2)$ case, as illustrated in 
Figure \ref{fig:nornoraccu} for a single $\mathcal{N}(0, 2)$ dataset.
 
\begin{figure}[!h]
\includegraphics[width=.5\textwidth]{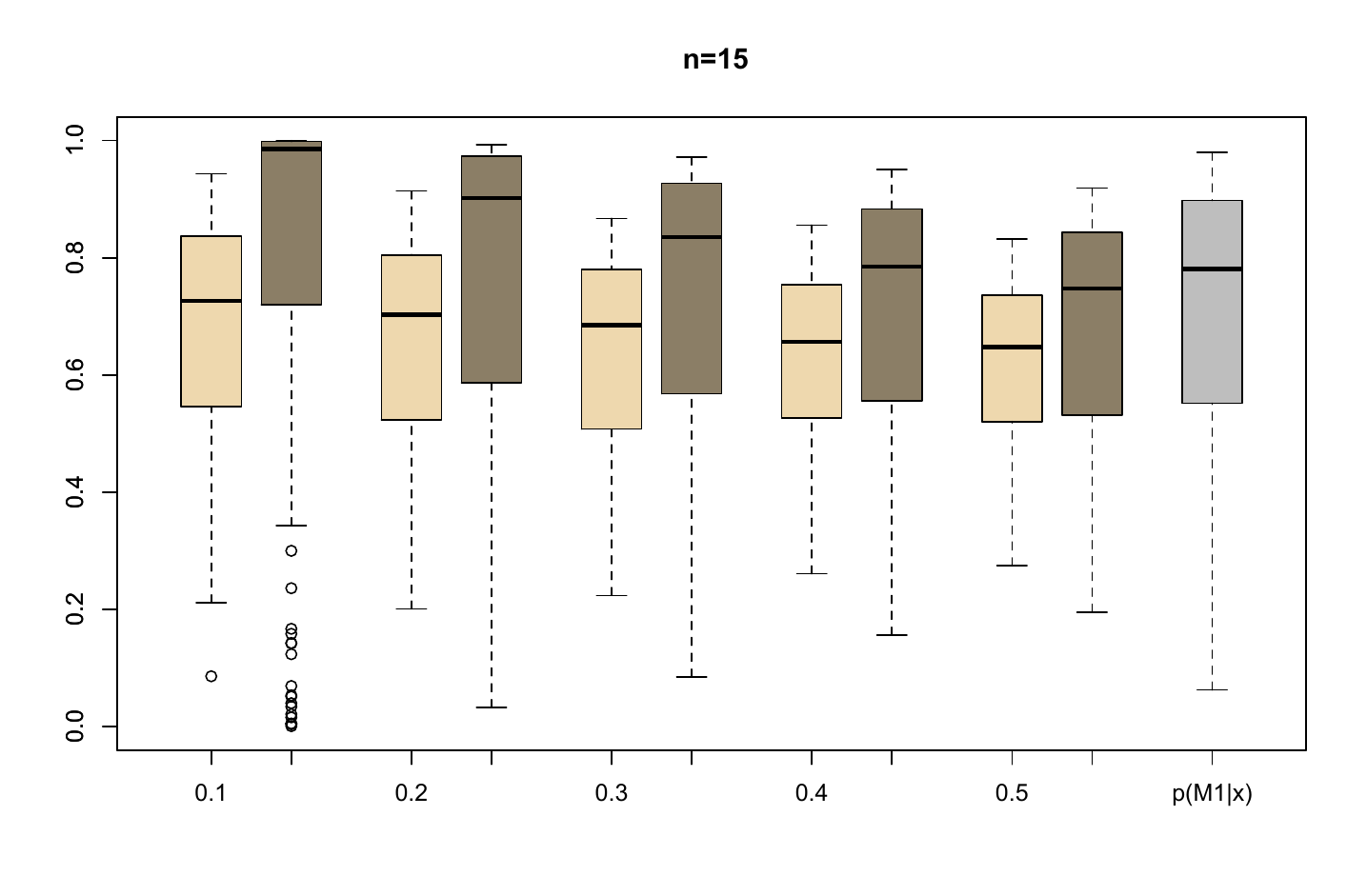}\includegraphics[width=.5\textwidth]{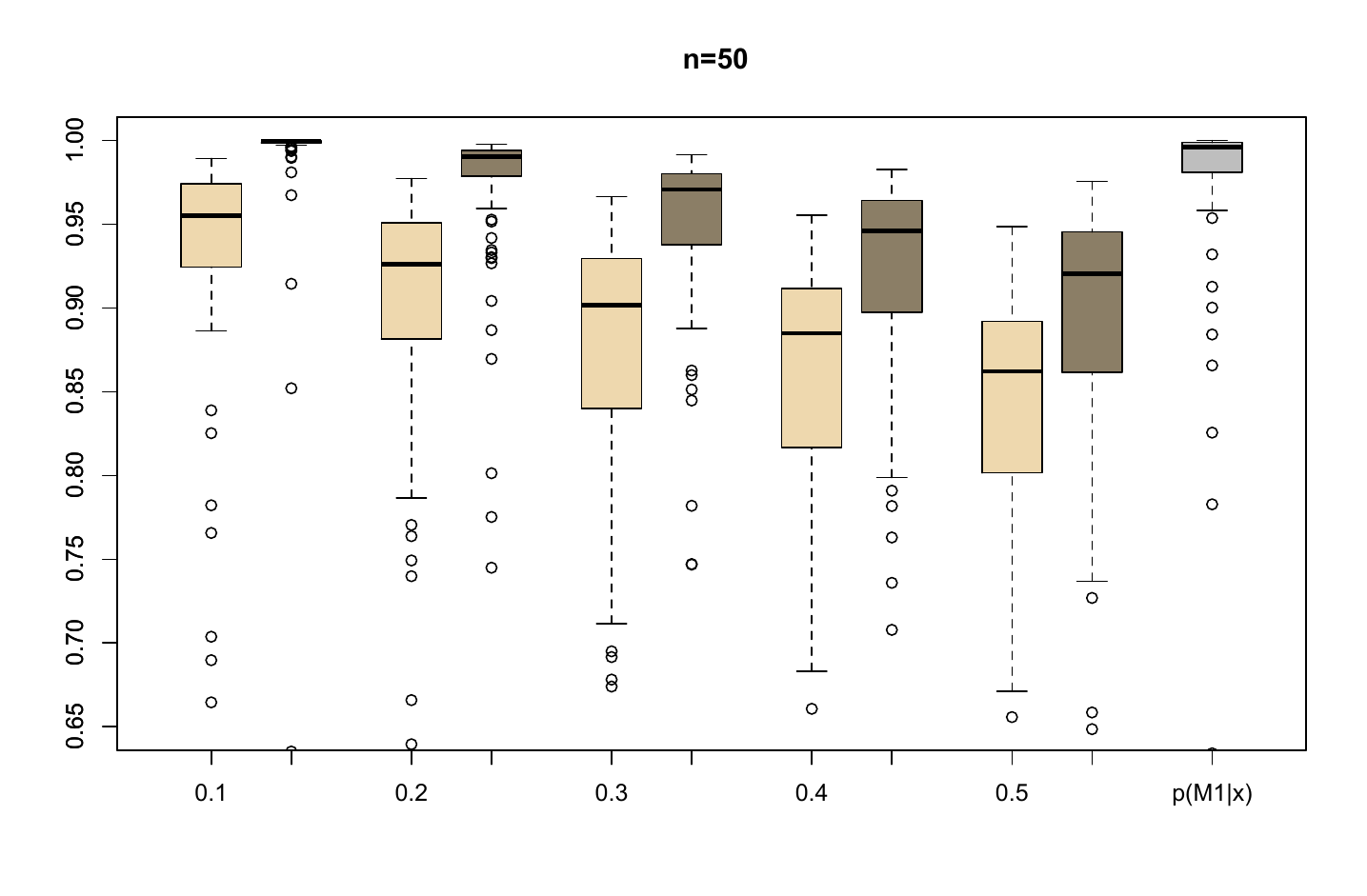}
\includegraphics[width=.5\textwidth]{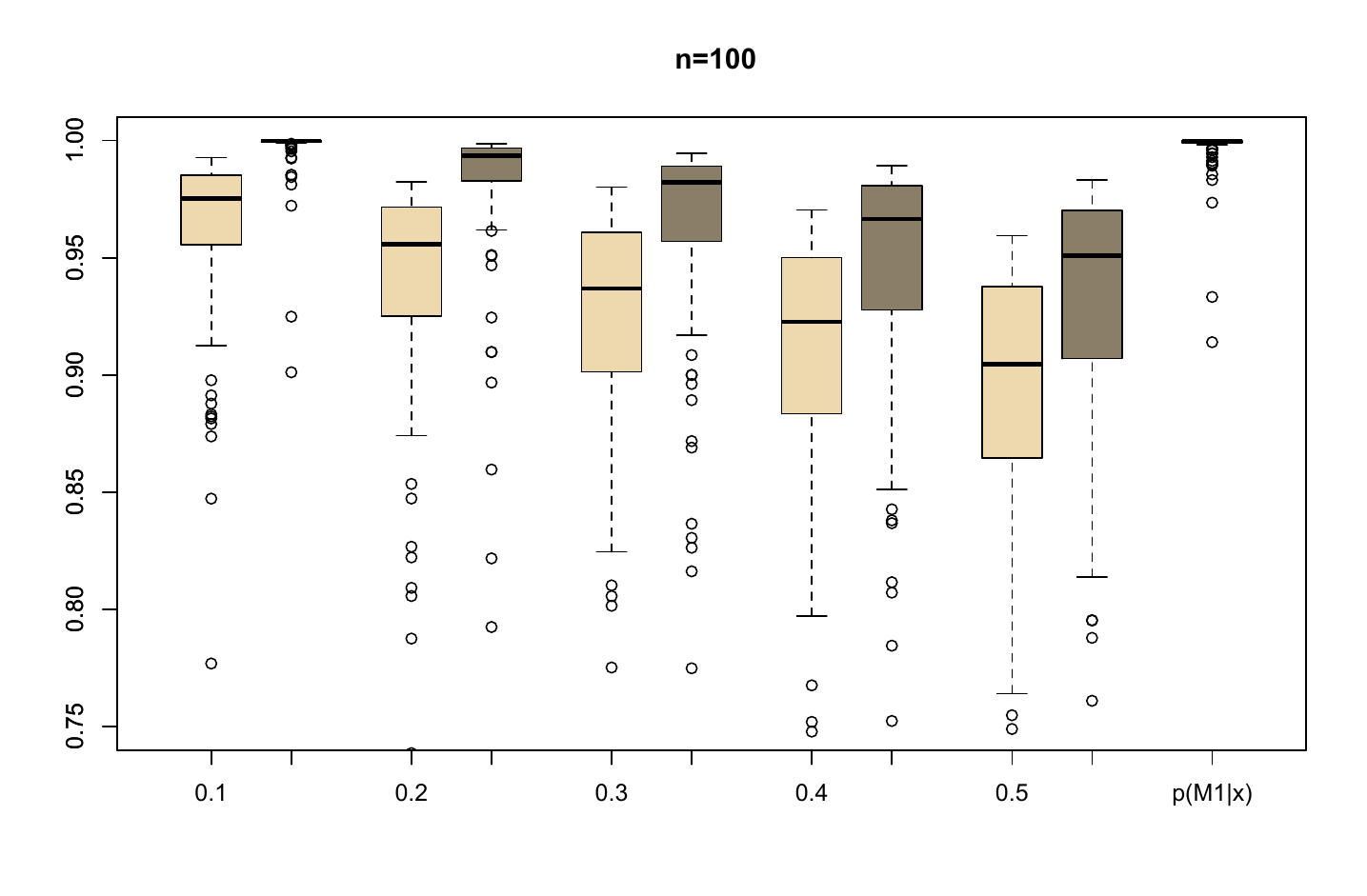}\includegraphics[width=.5\textwidth]{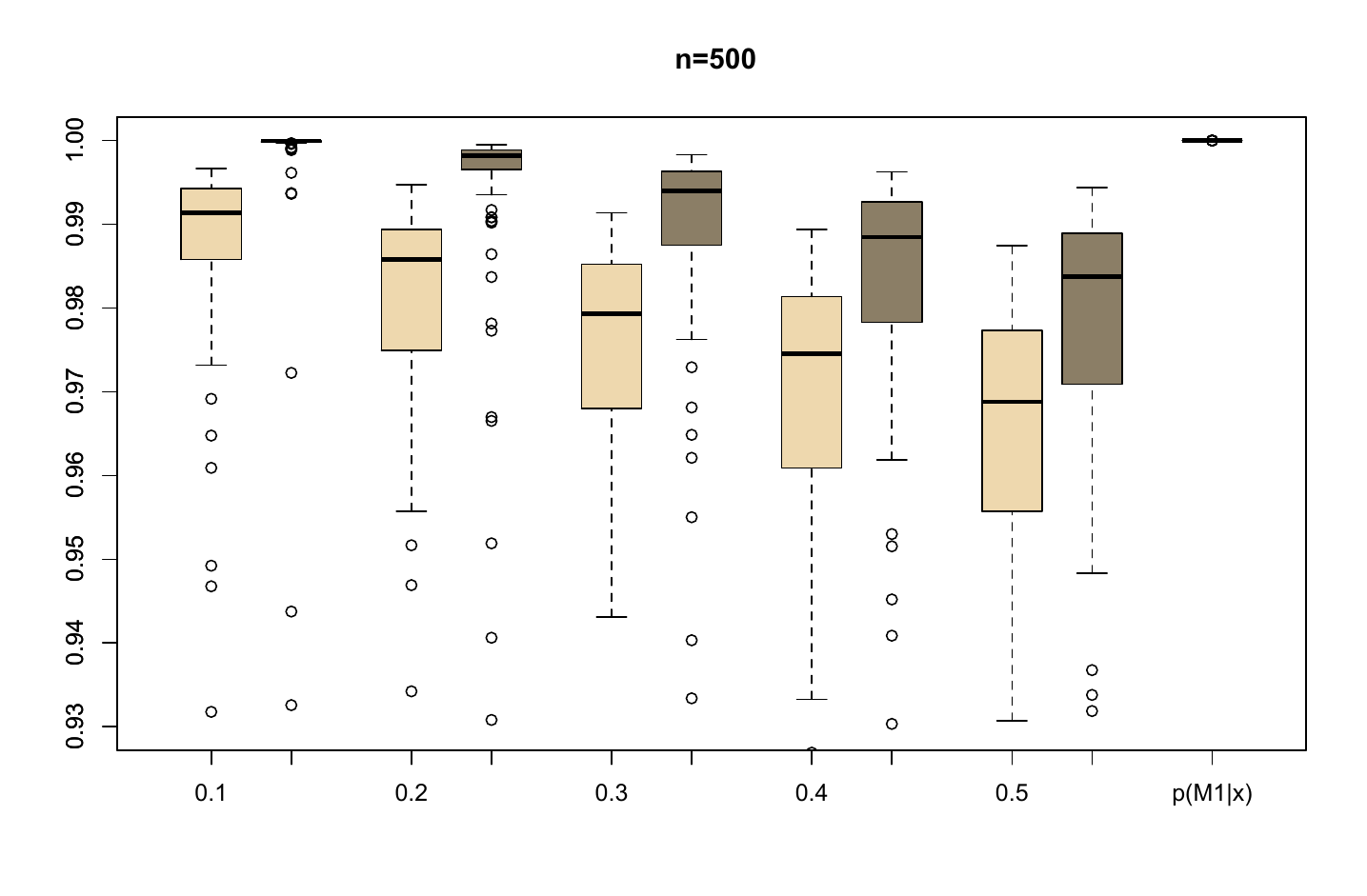}
\caption{\small {\bf Normal Example:} Boxplots of the posterior means {\em (wheat)} and medians of $\alpha$ {\em
(dark wheat)}, compared with a boxplot of the exact posterior probabilities of $\MF_0$ {\em (gray)} for a $\mathcal{N}(0, 1)$
sample, derived from 100 datasets for sample sizes equal to $15, 50, 100, 500$. Each posterior approximation is based on
$10^4$ MCMC iterations.}
\label{normb2}
\end{figure}

\begin{figure}
\includegraphics[width=.5\textwidth]{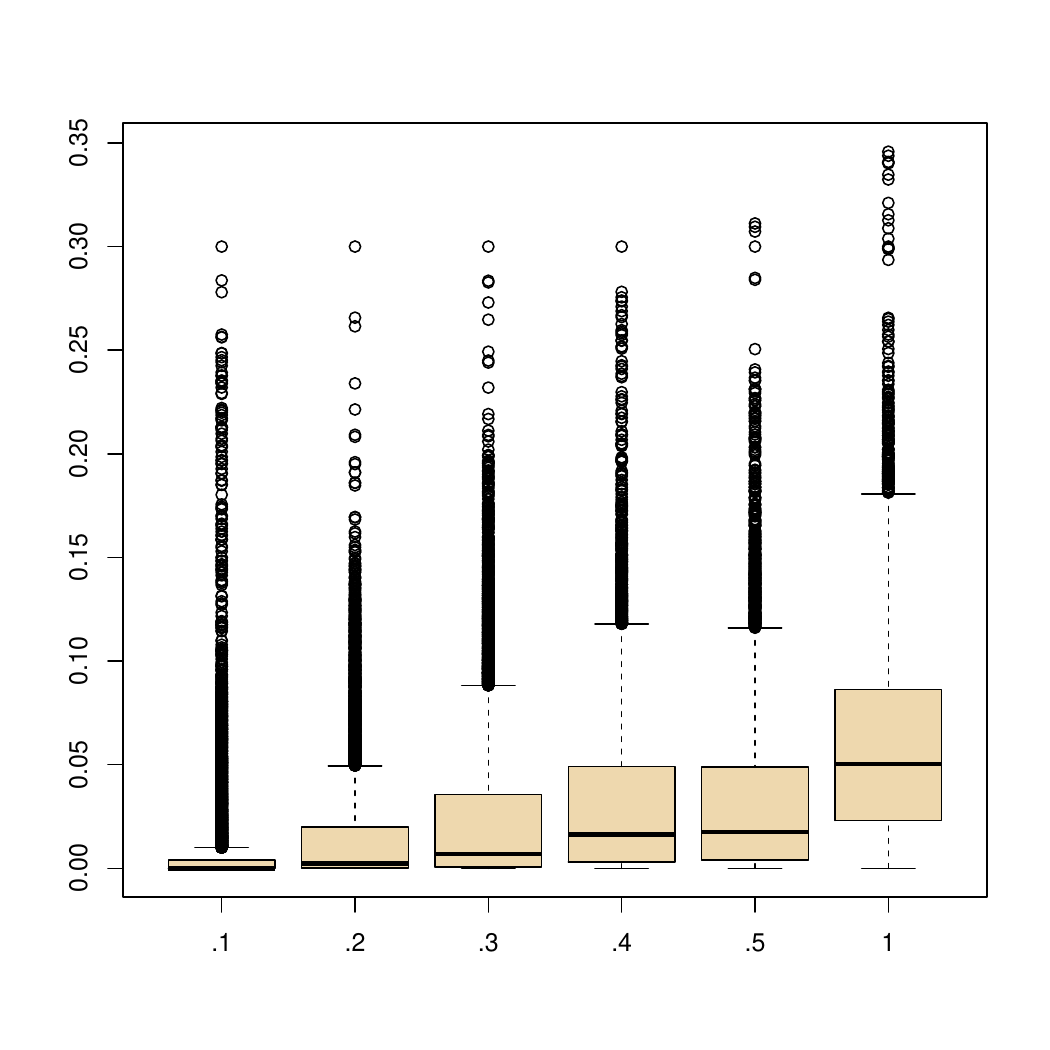}\includegraphics[width=.5\textwidth]{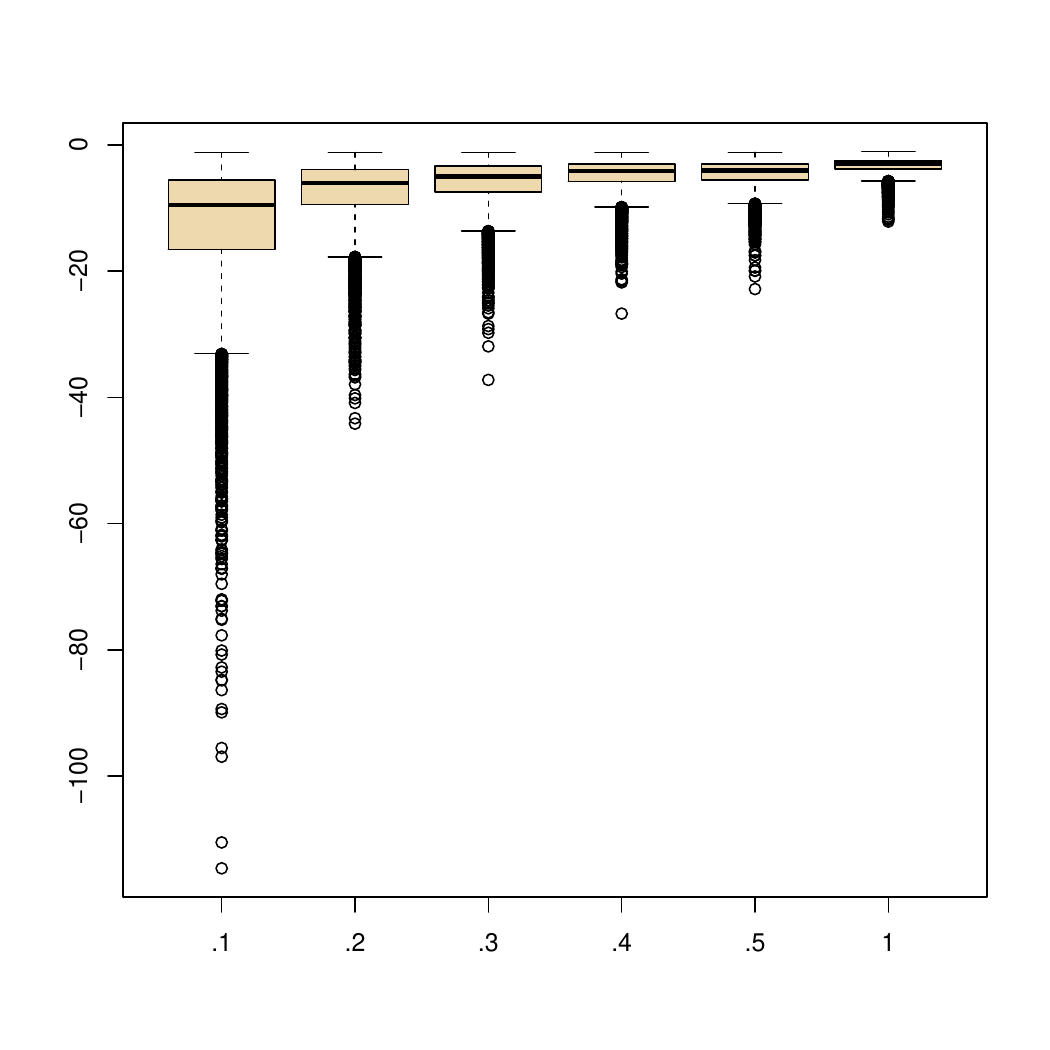}
\caption{\small \label{fig:nornoraccu}
{\bf Normal Example:} {\em (left)} Posterior distributions of the mixture weight $\alpha$ and {\em (right)} of
their logarithmic transform $\log\{\alpha\}$ under a Beta $\mathcal{B}(a_0,a_0)$ prior when $a_0=.1,.2,.3,.4,.5,1$ and
for a single Normal $\mathcal{N}(0,2)$ sample of $10^3$ observations. The MCMC outcome is based on $10^4$ iterations.}
\end{figure}

\section{Asymptotic consistency}
\label{sec:consix}

In this Section we study the asymptotic properties of our mixture testing procedure.  More precisely we study the asymptotic behaviour of the posterior distribution of $\alpha$ and we prove that the posterior on $\alpha$ concentrates on the \textit{true} value of $\alpha$ in the sense that if model $\model_1$ is correct the posterior distribution concentrates on $\alpha=1$, if model $\model_2$ is correct then the posterior distribution concentrates on $\alpha = 0$  and if neither are correct then the posterior concentrate on the value of $\alpha $ which minimizes the Kullback-Leibler divergence. This shows that the posterior on $\alpha$ leads to a consistent testing procedure. Moreover we also study the separation rate associated to such procedure when the models are embedded and we show that the approach leads to optimal separation rate, contrarywise to the Bayes factor which has an extra $\sqrt{\log n} $ factor.  
 
To do so we  consider two different cases. In the first case, the two models,
$\model_1$ and $\model_2$, are well separated while, in the second case, model $\model_1$ is a submodel of $\model_2$.
We denote by  $\pi$ the prior distribution on $(\alpha, \theta)$ with $\theta=(\theta_1, \theta_2)$ and assume that $\theta_j \in \Theta_j
\subset \R^{d_j}$. We first prove that, under weak regularity conditions on each model, we can obtain posterior
concentration rates for the marginal density $f_{\theta,\alpha}(\cdot) = \alpha f_{1,\theta_1} (\cdot)+ (1- \alpha)
f_{2,\theta_2}(\cdot)$. Let $\mathbf{x}^n = ( x_1, \cdots , x_n) $ be a $n$ sample with true density $f^*$.

\begin{proposition}\label{prop:cons}
Assume that, for all $C_1 >0$, there exist $\Theta_n $ a subset of $\Theta_1\times \Theta_2$  and $B_0, B_1\geq 0$ such that 
\begin{equation}\label{cond:thetan}
\pi\left[ \Theta_n^c\right] \leq n^{-C_1}, \quad \Theta_n \subset \{ \| \theta_1\| + \|\theta_2\| \leq B_0n^{B_1}\}
\end{equation}
and that there exist $H\geq 0$ and $L, \delta>0$ such that, for $j=1, 2$,
\begin{equation}
\begin{split}
& \sup_{\theta, \theta' \in \Theta_n} \| f_{j,\theta_j} - f_{j,\theta_j^{'}}\|_1 \leq Ln^H \| \theta_j - \theta_j'\|,
\quad \theta = (\theta_1, \theta_2) , \, \theta^{'} = (\theta_1^{'}, \theta_2^{'})\,, \\
& \forall \|\theta_j -\theta_j^*\|\leq \delta; \quad KL(f_{j,\theta_j}, f_{j,\theta_j^*})\lesssim  \| \theta_j -
\theta_j^*\|\,.
\end{split}\label{cond:Lip}
\end{equation}
We then have that, when $f^* = f_{\theta^*,\alpha^*}$, with $\alpha^*\in [0,1]$, there exists $M>0$ such that 
\begin{equation*}
\pi\left[ (\alpha, \theta); \|f_{\theta, \alpha} - f^*\|_1 > M \sqrt{ \log n/n}  | \mathbf{x}^n \right] =o_p(1) \,.
\end{equation*}
\end{proposition}
The proof of Proposition \ref{prop:cons} is a direct consequence of Theorem 2.1 of \citet{ghosal:ghosh:vdv:00} and is
thus omitted here.   Condition \eqref{cond:Lip} is a weak regularity condition on each of the candidate models.
Combined with condition \eqref{cond:thetan} it allows consideration of noncompact parameter sets in the usual way; see, for instance, 
\citet{ghosal:ghosh:vdv:00}. It is satisfied in all examples considered in this paper. We build on Proposition \ref{prop:cons} 
to describe the asymptotic behaviour of the posterior distribution on the parameters. It is possible to sharpen the above posterior 
concentration rate into $M_n/\sqrt{n}$ for any sequence  $M_n$ going to infinity by controlling the local entropy and obtaining precise 
upper bounds on neighbourhoods of  $f^*$. This is not useful in the case of separated models but becomes more important in the context 
of embedded models. Although we do not treat this here, following \cite{kleijn:vdv}, if the true distribution $f_0$ does not belong to the embedding model $f_{\theta,\alpha}$, then the posterior will concentrate on $f^*$ which minimizes the Kullback-Leibler divergence between $f_0$ and $f_{\theta,\alpha}$, at a similar rate. 

\subsection{The case of separated models}\label{subsec:asy:sep}

Assume that both models are separated in the sense that there is identifiability: 
\begin{equation}\label{ident}
\forall \alpha, \alpha'\in [0,1], \quad \forall \theta_j, \theta_j^{'}, \, j=1, 2\quad 
P_{\theta,\alpha} =P_{\theta^{'},\alpha^{'}}  \quad \Rightarrow \alpha = \alpha^{'}, \quad \theta = \theta^{'},
\end{equation}
where $P_{\theta, \alpha}$ denotes the distribution associated with $f_{\theta, \alpha}$. We assume that \eqref{ident}
also holds on the boundary of $\Theta_1\times \Theta_2$. In other words, the following
$$
\inf_{\theta_1\in \Theta_1}\inf_{\theta_2 \in \Theta_2} \| f_{1,\theta_1} - f_{2,\theta_2}\|_1 >0
$$
holds.  We also assume that, for all $\theta_j^* \in \Theta_j$, $j=1,2$,  if $P_{\theta_j} $ converges in the weak
topology to $P_{\theta_j^*}$, then $\theta_j $ converges in the Euclidean topology to $\theta_j^*$. The following result
then holds:

\begin{theorem} \label{cons:para}
Assume that \eqref{ident} is satisfied, together with \eqref{cond:thetan} and \eqref{cond:Lip}, then for all $\epsilon>0$ 
\begin{equation*}
\pi\left[ |\alpha - \alpha^*|>\epsilon  | \mathbf{x}^n \right] =o_p(1) .
\end{equation*}
In addition, assume that the mapping $\theta_j \rightarrow f_{j,\theta_j}$ is twice continuously differentiable in a
neighbourhood of $\theta_j^*$, $j =1,2$, and  that 
$$
f_{1,\theta_1^*}- f_{2,\theta_2^*}, \nabla f_{1,\theta_1^*}, \nabla f_{2,\theta_2^*}
$$
are linearly independent as functions of $y$ and that there exists $\delta>0$ such that
$$
\nabla f_{1,\theta_1^*}, \, \nabla f_{2,\theta_2^*}, \, \sup_{|\theta_1-\theta_1^*|<\delta}| D^2f_{1,\theta_1}|, \,
\sup_{|\theta_2-\theta_2^*|<\delta}| D^2f_{2,\theta_2}| \, \in L_1\,.
$$
Then  
\begin{equation} \label{rate}
\pi\left[ |\alpha - \alpha^*|>M\sqrt{\log n/n}  \big| \mathbf{x}^n \right] =o_p(1) .
\end{equation} 
\end{theorem}
Theorem \ref{cons:para} allows for the interpretation of the quantity $\alpha$ under the posterior distribution. In
particular, if the data $\mathbf x^n$ are generated from model $\model_1$ (resp. $\model_2$), then the posterior
distribution on $\alpha$ concentrates around $\alpha=1$ (resp. around $\alpha=0$), which establishes the consistency of
our mixture approach.

We now consider the embedded case. 

\subsection{Embedded case}

In this Section we assume that $\model_1 $ is a submodel of $\model_2$, in the sense that $\theta_2 = (\theta_1, \psi)$
with $\psi \in \mathcal S \subset \R^{d_\psi}$ and that $f_{2, \theta_2} \in \model_1$ when $\theta_2 = (\theta_1, \psi_0)$ for
some given value $\psi_0$, say $\psi_0=0$. Condition \eqref{ident} is no longer verified for all $\alpha$'s: we assume
however that it is verified for all $\alpha, \alpha^*\in [0,1)$ and that $\theta_2^*  = (\theta_1^*,\psi^*)$ satisfies
$\psi^*  \neq 0$. In this case, under the same conditions as in Theorem \ref{cons:para}, we immediately obtain the
posterior concentration rate $\sqrt{\log n/n}$ for estimating $\alpha$ when $\alpha^*\in [0,1)$ and $\psi^*\neq 0$ and Theorem \ref{cons:para} implies that \eqref{rate} holds, which in turns implies that if $\alpha^* = 1$, i.e. if the distribution comes from model $\model_2$, 
$$ \pi\left[  \alpha >M\sqrt{\log n/n}  | \mathbf{x}^n \right] =o_p(1). $$ 

We now treat the case where $\psi^*=0$; in other words, $f^*$ is in model $\model_1$.  

As in \citet{rousseau:mengersen:2011}, we consider both possible paths to approximate $f^*$: either $\alpha$ goes to
1 or $\psi$ goes to $\psi_0=0$. In the first case, called path 1, $(\alpha^*,\theta^*) = (1, \theta_1^*, \theta_1^*, \psi)$ 
with $\psi \in \mathcal S$; in the second, called path 2, $(\alpha^*, \theta^*) = (\alpha, \theta_1^*, \theta_1^*, 0)$ with 
$\alpha \in [0,1]$. In either case, we write $P^*$ as the distribution and denote $F^*g = \int f^*(x) g(x) d\mu(x) $ for any
integrable function $g$. For sparsity reasons, we consider the following structure for the prior on $(\alpha, \theta)$: 
$$
\pi( \alpha, \theta) = \pi_\alpha(\alpha) \pi_1( \theta_1) \pi_\psi(\psi) , \quad \theta_2 =  (\theta_1, \psi).
$$
This means that the parameter $\theta_1$ is common to both models, i.e., that $\theta_2$ shares the parameter $\theta_1$
with $f_{1, \theta_1}$. 

Condition \eqref{ident} is replaced by 
\begin{equation} \label{identbis}
P_{\theta, \alpha} = P^* \quad \Rightarrow \alpha = 1, \quad \theta_1 = \theta_1^*, \quad \theta_2 = (\theta_1^*, \psi) \quad \mbox{or } \quad \alpha \leq 1 , \quad \theta_1 = \theta_1^*, \quad \theta_2 = (\theta_1^*, 0)
\end{equation}
Let $\Theta^*$ be the above parameter set.  
 
As in the case of separated models, the posterior distribution concentrates on $\Theta^*$. We now describe more
precisely the asymptotic behaviour of the posterior distribution, using \cite{rousseau:mengersen:2011}. We cannot apply
directly Theorem 1 of \cite{rousseau:mengersen:2011}, hence the following result is an adaptation of it. We require the
following assumptions with $f^* = f_{1, \theta_1^*}$. For the sake of simplicity, we assume that $\Theta_1$ and
$\mathcal S$ are compact. Extension to non compact sets can be handled similarly to \cite{rousseau:mengersen:2011}. 

\begin{itemize}
\item [B1] \textit{Regularity}: 
Assume that $\theta_1\rightarrow  f_{1,\theta_1}$ and $\theta_2 \rightarrow f_{2, \theta_2}$  are 3 times continuously differentiable and that 
 \begin{equation*}
 F^*\left( \frac{ \bar f_{1, \theta_1^*}^3 }{ \underline f_{1, \theta_1^*}^3   } \right)<+\infty , \quad \bar f_{1, \theta_1^*} = \sup_{|\theta_1- \theta_1^*| <\delta} f_{1, \theta_1}, \quad  \underline f_{1, \theta_1^*} = \inf_{|\theta_1- \theta_1^*| <\delta} f_{1, \theta_1}
 \end{equation*}
 \begin{equation*}
 \begin{split}
 & F^*\left( \frac{ \sup_{|\theta_1- \theta_1^*| <\delta} |\nabla f_{1, \theta_1^*} |^3  }{ \underline f_{1, \theta_1^*}^3  } \right)<+\infty ,\quad  F^*\left( \frac{  |\nabla f_{1, \theta_1^*} |^4  }{ f_{1, \theta_1^*}^4  } \right)<+\infty ,\\
 & F^*\left( \frac{ \sup_{|\theta_1- \theta_1^*| <\delta} |D^2 f_{1, \theta_1^*} |^2  }{ \underline f_{1, \theta_1^*}^2  } \right)<+\infty , \quad    F^*\left( \frac{ \sup_{|\theta_1- \theta_1^*| <\delta} |D^3 f_{1, \theta_1^*} |  }{ \underline f_{1, \theta_1^*}  } \right)<+\infty
 \end{split}
 \end{equation*}

\item [B2] \textit{Integrability}: There exists $\mathcal S_ 0 \subset \mathcal S \cap\{ |\psi| > \delta_0\}$, for some positive $\delta_0$ and  satisfying $\mbox{Leb}(\mathcal S_0) >0$, and such that for all $\psi \in \mathcal S_0$,
$$F^*\left( \frac{ \sup_{|\theta_1-\theta_1^*|<\delta} f_{2, \theta_1,\psi}}{ f_{1,\theta_1^*}^4}\right) < +\infty,  \quad F^*\left( \frac{\sup_{|\theta_1-\theta_1^*|<\delta} f_{2, \theta_1,\psi}^3 }{  \underline{f}_{1,\theta1^*}^3 } \right) <+\infty,$$

\item [B3] \textit{Stronger identifiability} : 
Set 
$$ 
\nabla f_{2, \theta_1^*,\psi^*}(x) = \left( \nabla_{\theta_1} f_{2, \theta_1^*,\psi^*}(x)^{\text{T}} , \nabla_\psi f_{2, \theta_1^*,\psi^*}(x)^{\text{T}}\right)^{\text{T}} .
$$ 
Then for all $\psi \in \mathcal S$ with  $\psi\neq 0$, if $\eta_0 \in \R$, $\eta_1 \in \R^{d_1} $
\begin{equation} \label{cond:strong:ident}
 \eta_0 ( f_{1, \theta_1^*} - f_{2, \theta_1^*, \psi} ) + \eta_1^{\text{T}}[ \nabla_{\theta_1}f_{1, \theta_1^*} -  \nabla_{\theta_1} f_{2, \theta_1^*,\psi}(x)] = 0 \quad \Leftrightarrow \eta_1 = 0, \, \eta_2 = 0 
\end{equation}
  
\end{itemize}

Assumptions B1-B3 are  similar, but weaker, to \cite{rousseau:mengersen:2011}'s set of conditions and in fact B3 is milder than the strong identifiability condition imposed in that paper. Hence these conditions are satisfied for a wide range of regular models.

We can now state the main theorem:
\begin{theorem}\label{th:cons:embed}
Given the model
$$
f_{\theta_1, \psi, \alpha} = \alpha f_{1, \theta_1} + (1- \alpha)f_{2, \theta_1, \psi}, 
$$
assume that the data comprise the $n$ sample $\mathbf x^n =(x_1,\cdots, x_n)$ issued from $f_{1, \theta_1^*}$ for some
$\theta_1^*\in \Theta_1$, and that assumptions $B1-B3$ are satisfied. Then for all sequence $M_n$ going to infinity,
\begin{equation} \label{rate:sharp}
\pi\left[ (\alpha, \theta); \|f_{\theta, \alpha} - f^*\|_1 > M_n /\sqrt{ n}  | \mathbf{x}^n \right] =o_p(1) .
\end{equation}
If the prior $\pi_\alpha$ on $\alpha$ is a Beta $\mathcal{B}(a_1,a_2)$ distribution, with $a_2 < d_\psi$, and if the prior
$\pi_{1}\pi_{\psi}$ is absolutely continuous with positive and continuous density at $(\theta_1^*, 0)$, then for all
$M_n$ going to infinity, 
$$
\pi\left[ |\alpha - 1 |>M_n /\sqrt{n} | \mathbf{x}^n \right] =o_p(1) .
$$
If $a_2> d_\psi$, then for any $e_n= o(1)$, 
$$
\pi\left[ |\alpha - 1 | < e_n | \mathbf{x}^n \right] =o_p(1) . 
$$

\end{theorem}

Note that  the phase transition on the behaviour of the posterior distribution is $a_2<d_\psi$ versus $a_2>d_\psi$, which is not quite the same as in \cite{rousseau:mengersen:2011}. 

Theorems \ref{th:cons:embed} and \ref{cons:para} imply that testing decisions can be taken based on the posterior distribution of $1-\alpha$ when $a_2<d_\psi$.
Indeed, in this case if one considers a testing approach of the form: $H_0$ is rejected if $\pi(1-\alpha > M_n /\sqrt{n} | \mathbf x^n )\geq 1/2$
for some sequence $M_n$ large or increasing to infinity, then this testing procedure is consistent under both the null and the alternative. 

In contrast to the Bayes factor which converges to 0 under the alternative model $\model_2$ exponentially quickly, 
the convergence rate of $\alpha$ to $\alpha^*\neq 1$ is of order $1/\sqrt{n}$. However this does not mean that the separation rate 
of the procedure based on the mixture model is worse than that of the Bayes factor. On the contrary, while it is well known that the Bayes factor leads to a separation rate of order $\sqrt{\log n}/\sqrt{n}$ in parametric models, we show in the following theorem that our approach can lead to a testing procedure with a better separation rate of order $1/\sqrt{n}$. 

To prove the following result we need to strengthen slightly assumption B3:\\ 
 B4 \textit{second order identifiability condition} : \\
Set 
$  D_\psi^2 f_{2,\theta_1,0}$ as the second derivative of $f_{2,\theta_1,\psi}$ with respect to $\psi$ calculated at $\theta = (\theta_1,0)$ 
Then for all $\theta_1 \in \Theta_1$ if  $\eta_1 \in \R^{d_1}, \eta_2, \eta_3 \in \mathbb R^{d_\psi} $
\begin{equation} \label{cond:strong:identbis}
\begin{split}
 \eta_1^{\text{T}}\nabla_{\theta_1}f_{1, \theta_1} + \eta_2^{\text{T}}\nabla_\psi f_{2, \theta_1,0}(x) + \eta_3^{\text{T}}D_\psi^2 f_{2,\theta_1,0} \eta_3= 0 \quad \Leftrightarrow \eta_1 = 0, \, \eta_2 =\eta_3 &= 0\\
\end{split}
\end{equation}

Note that condition B4 is very similar to the strong identifiability condition of \cite{rousseau:mengersen:2011}. 

\begin{theorem}\label{th:separationrate}
Given the model
$$
f_{\theta, \alpha} = f_{\theta_1, \psi, \alpha} = \alpha f_{1, \theta_1} + (1- \alpha)f_{2, \theta_1, \psi}, \quad \theta = (\theta_1, \psi)
$$
assume that the data comprise the $n$ sample $\mathbf x^n =(x_1,\cdots, x_n)$ issued from $f_n^* = f_{2, \theta_{1,n}, \psi_n}$ for some
 some sequence $\theta_{1,n} \in \Theta_1$ and $\psi_n \in \mathcal S$ Let assumptions $B1-B4$ be satisfied.
Moreover if the prior
$\pi_{\theta_1, \psi}$ is absolutely continuous with positive and continuous density on $\Theta$ and if the prior $\pi_\alpha$ on $\alpha$ is a Beta $\mathcal{B}(a_1,a_2)$ distribution then there exists $M'>0$ such that 
$$
\sup_{\theta_{1,n} \in \Theta_1, \|\psi_n\| \geq M_n/\sqrt{n}} E_{\theta_{1,n},\psi_n} \pi\left[ |\alpha - 1 | \leq M' M_n^2/\sqrt{n} | \mathbf{x}^n \right] =o(1) 
$$
for any sequence $M_n $ going to infinity such that $M_n^2 = o(\sqrt{n})$. 

\end{theorem}

Theorem \ref{th:separationrate} implies in particular that if the testing
procedure is: $H_0$ is rejected as soon as $\pi(1-\alpha > M_0
/\sqrt{n}|\mathbf x^n)\leq 1/2 $ with $M_0$ an arbitrarily large constant then
the separation rate is of order $\sqrt{M_0}/\sqrt{n}$.  Although Theorem
\ref{th:separationrate} holds for any value of $a_2$ and $d_\psi$, for the
testing procedure to make sense one needs to choose $a_2 < d_\psi$, since,
otherwise, for any $e_n = o(1)$, the posterior distribution $\pi(1-\alpha >
e_n|\mathbf x^n)= o_p(1)$ under $H_0$. Calibrating the procedure by a prior
predictive approach under both  $H_0$ and $H_1$ will lead to a consistent
testing procedure.


\section{Illustrations}\label{sec:firsT}

In this Section, we present three further examples that demonstrate the performance of the mixture estimation approach
and provide confirmation of the consistency results obtained in Section \ref{sec:consix}. 
The first follows from the example given in Section \ref{sec:eg1} and is a direct application of Theorem \ref{cons:para}.
The second is cast in a nonparametric setting and is an application of Theorem \ref{th:cons:embed}.
The third example is a case study that illustrates the hypothesis testing approach in a regression setup.

\vspace{0.3cm}
\begin{example}\label{ex:NoLap}
Inspired by \cite{marin:pillai:robert:rousseau:2014}, we oppose the Normal
$\mathcal{N}(\mu,1)$ model to the double-exponential
$\mathcal{L}(\mu,\sqrt{2})$ model. The scale $\sqrt{2}$ is intentionally chosen
to make both distributions share the same variance. As in the Normal case of
Section \ref{sec:eg1}, the location parameter $\mu$ can be shared by both
models and allows for the use of the flat Jeffreys' prior. As in the example in
Section \ref{sec:eg1}, Beta distributions $\mathcal{B}(a_0,a_0)$ are compared
with respect to their hyperparameter $a_0$. However, whereas in the previous
example we illustrated that the posterior distribution of the weight of the
true model converged to $1$, we now consider a setting in which neither model
is correct. We achieve this feature by using a
$\mathcal{N}(0, .7^2)$ distribution to simulate the data as it corresponds to
neither model $\mathfrak{M}_1$ nor to model $\mathfrak{M}_2$. In this specific
case, both posterior means and medians of $\alpha$ fail to concentrate near $0$
and $1$ as the sample size increases, as shown in Figure \ref{norlap1}. Thus in
the majority of cases in this experiment, the outcome indicates that neither of
both models is favored by the data. This example does not exactly follow the
assumptions of Theorem \ref{cons:para} since the Laplace distribution is not
differentiable everywhere. However, it is both almost surely differentiable and
differentiable in quadratic mean, so we expect to see the same types of
behaviour as predicted by Theorem \ref{cons:para}.

\begin{figure}[!h]
\includegraphics[width=.33\textwidth]{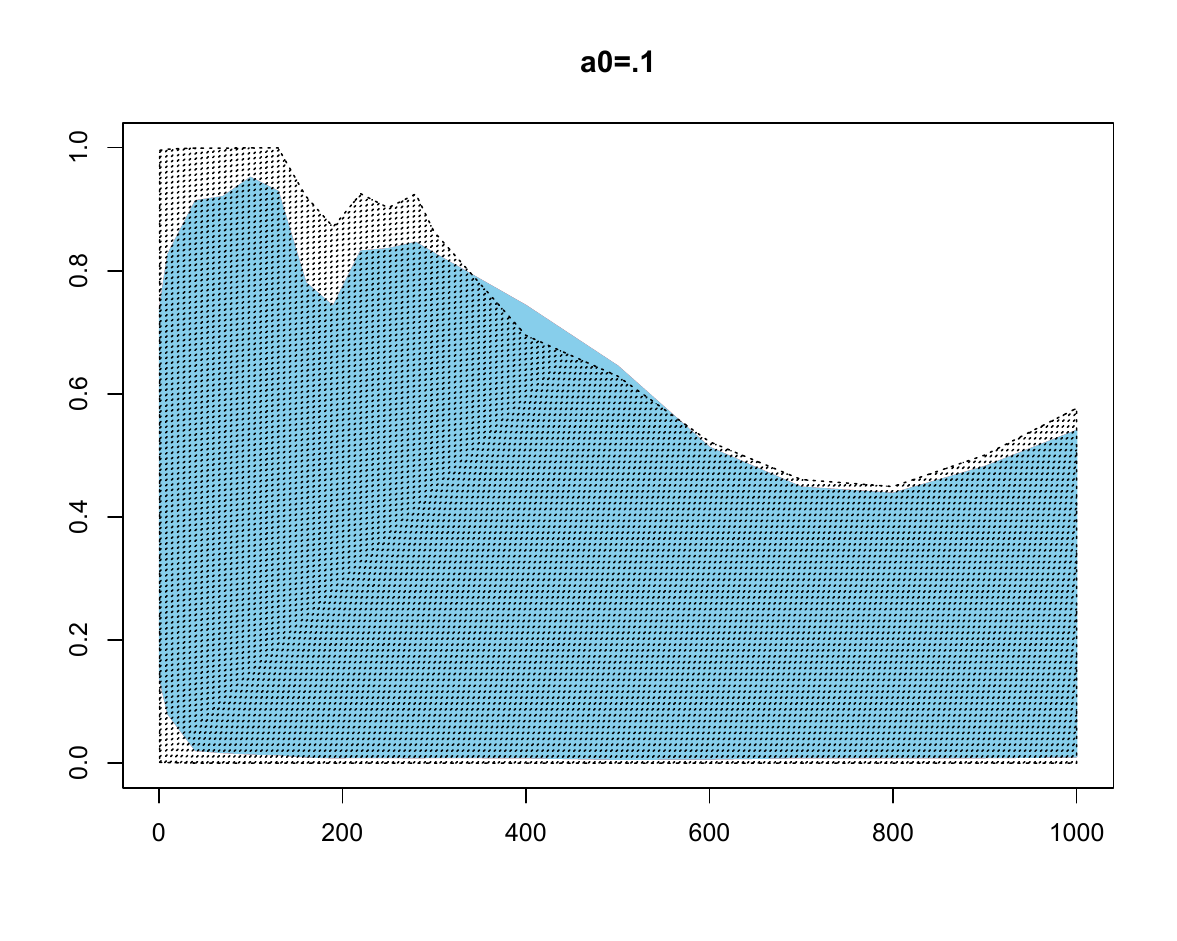}\includegraphics[width=.33\textwidth]{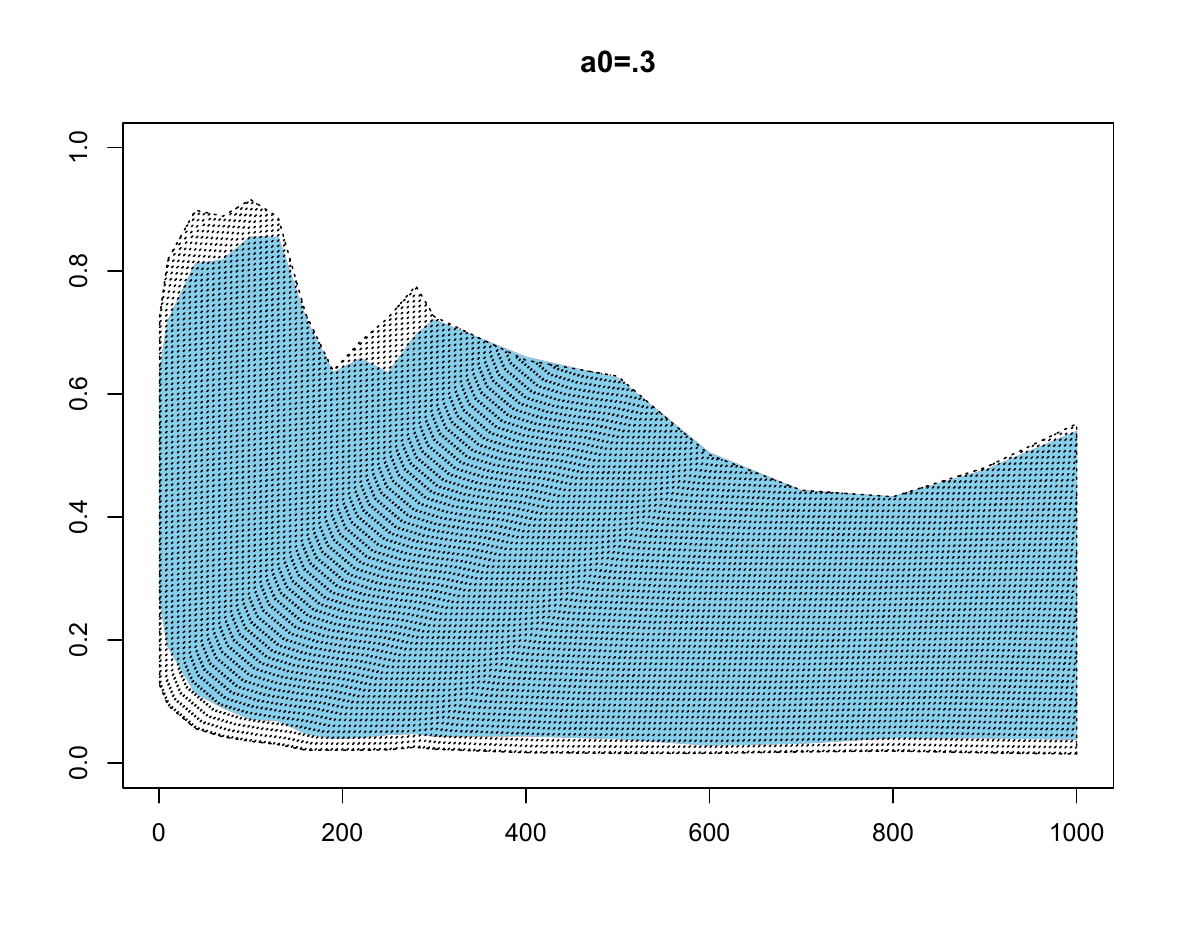}\includegraphics[width=.33\textwidth]{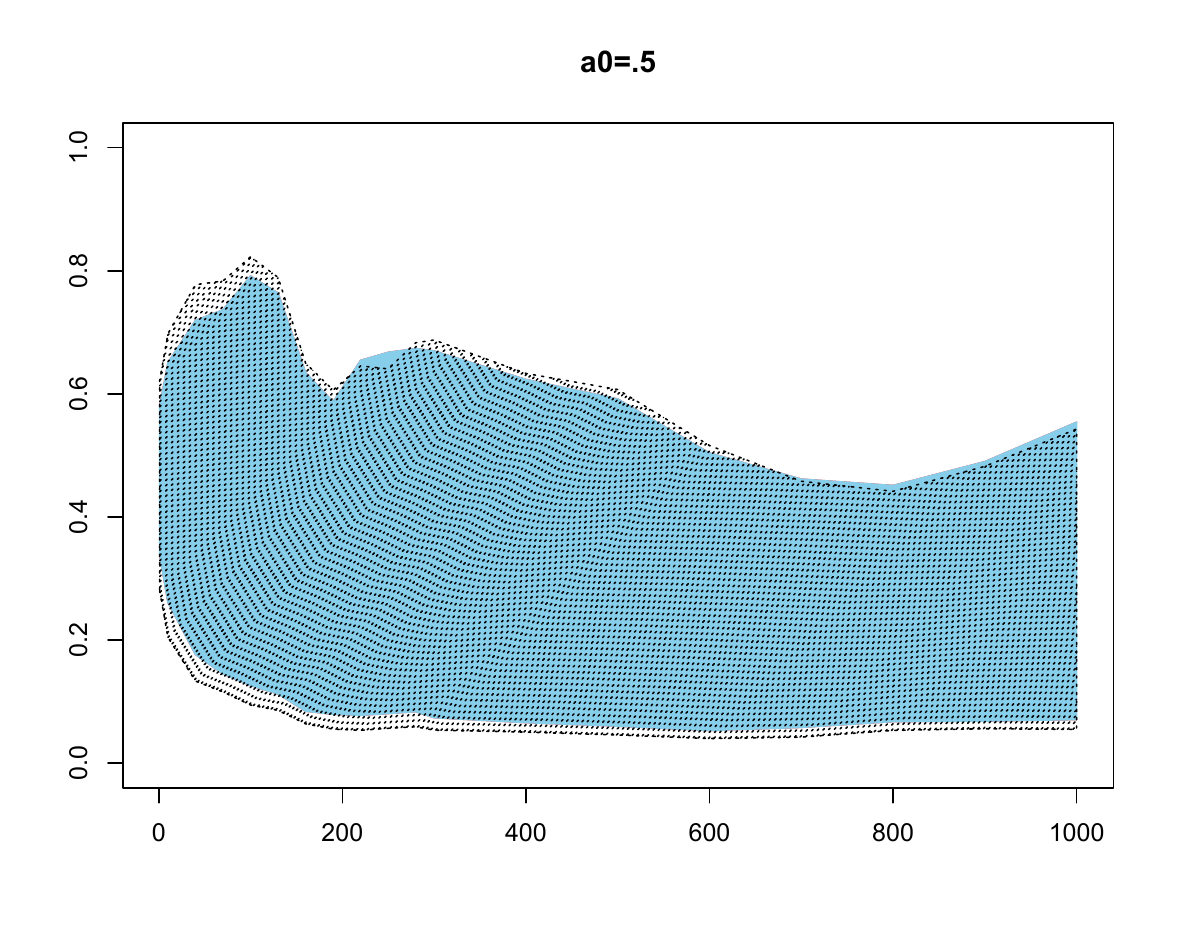}
\caption{\small {\bf Example \ref{ex:NoLap}:} Ranges of posterior means {\em(skyblue)} and medians {\em(dotted)} of
the weight $\alpha$ of model $\mathcal{N}(\theta,1)$ over 100 $\mathcal{N}(0, .7^2)$ datasets for sample sizes from 1 to 1000. 
Each estimate is based on a Beta prior with $a_0=.1, .3, .5$, respectively, and $10^4$ MCMC iterations.}
\label{norlap1}
\end{figure}

In this example, the Bayes factor associated with Jeffreys' prior is defined as
$$
\mathfrak{B}_{12}=\dfrac{\exp\left\{ \nicefrac{-\sum_{i=1}^n(x_i-\bar{x})^2}{2}\right\}}{(\sqrt{2 \pi})^{n-1}\sqrt{n}}\Big/
\int_{-\infty}^{\infty}\dfrac{\exp\left\{\nicefrac{-\sum_{i=1}^n|x_i-\mu|}{\sqrt{2}}\right\}}{(2\sqrt{2})^n}\dd\mu\,
$$
where the denominator is available in closed form.
As above, since the prior is improper, it is formally
undefined, even though the classical Bayesian approach argues in favour of
using the same prior on both $\mu$'s. Nonetheless, we employ it in order to
compare Bayes estimators of $\alpha$ with the posterior
probability of the model being a $\mathcal{N}(\mu, 1)$ distribution. Based on a Monte Carlo experiment involving 100
replicas of a $\mathcal{N}(0, .7^2)$ dataset, Figure \ref{norlap2} demonstrates the reluctance of the estimates of $\alpha$
to approach $0$ or $1$, while $\mathbb{P}(\MF_1|\bx)$ varies over the whole
range between $0$ and $1$ for all sample sizes considered here.  While this is
a weakly informative indication, the right hand side of Figure \ref{norlap2}
shows that, on average, the posterior estimates of $\alpha$ converge toward a
value between $.1$ and $.4$ for all $a_0$ while the posterior probabilities
converge to $.6$. In this respect, both criteria offer a similar interpretation
about the data because neither $\alpha$ nor $P(\MF_1|x)$ provide definitive support for either model.

\begin{figure}[!h]
\includegraphics[width=.59\textwidth]{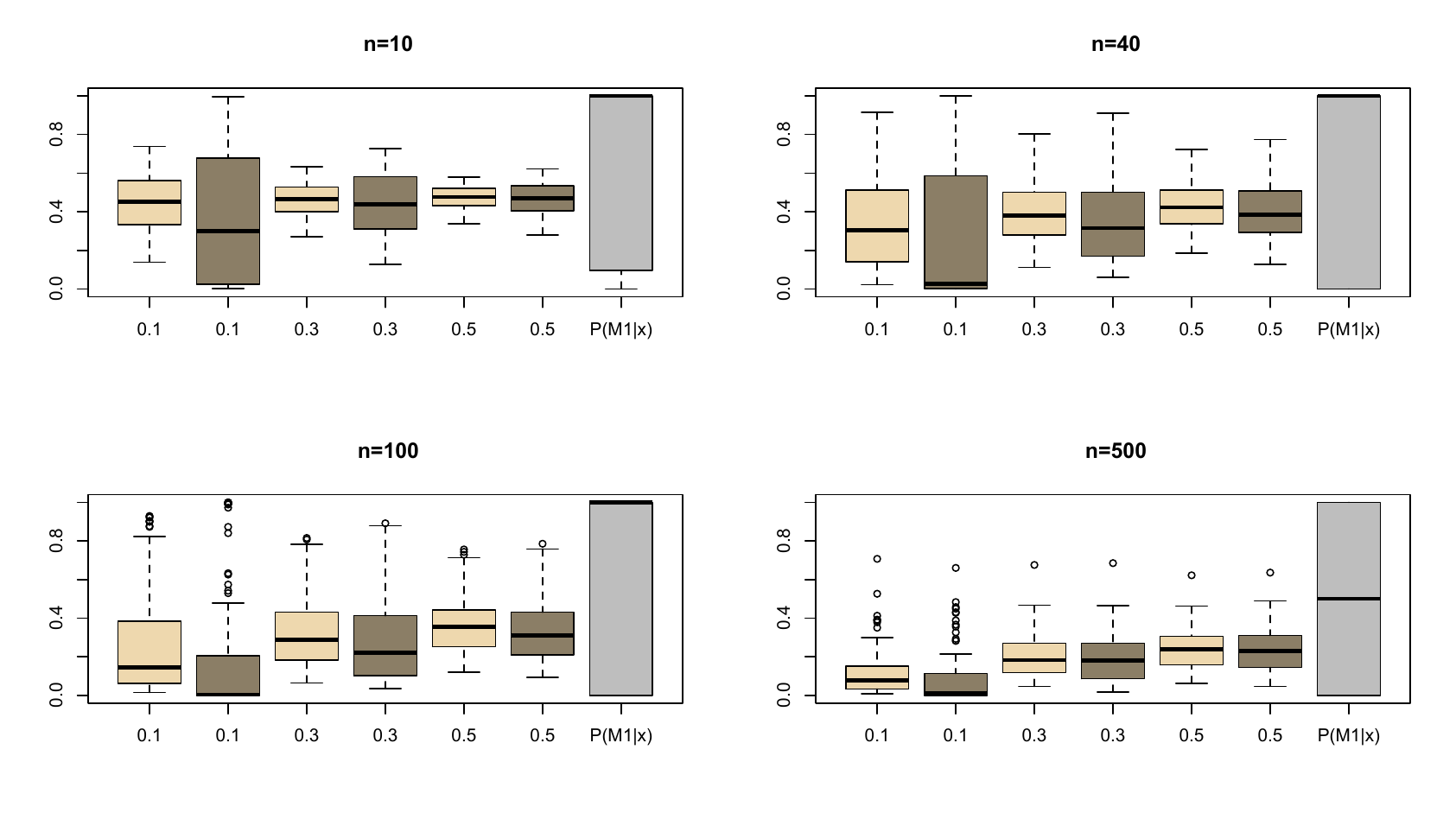}\includegraphics[width=.41\textwidth]{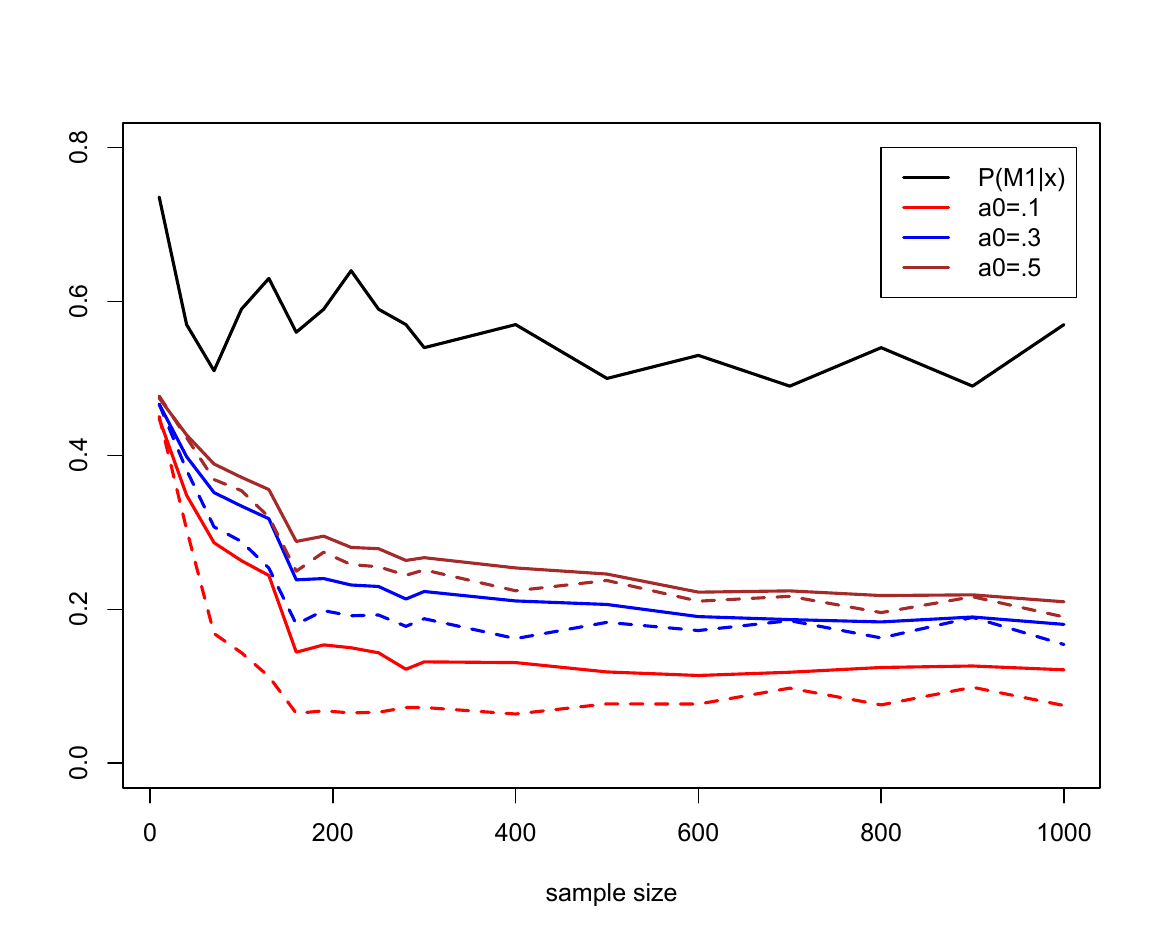}
\caption{\small {\bf Example \ref{ex:NoLap}:} {\em (left)} Boxplot of the posterior means {\em(wheat)} and medians {\em(dark wheat)} of
$\alpha$, and of the posterior probabilities of model $\mathcal{N}(\mu, 1)$ over 100 $\mathcal{N}(0, .7^2)$ datasets for
sample sizes $n=10, 40, 100, 500$; {\em (right)} averages of the posterior means and posterior medians of $\alpha$
against the posterior probabilities $\mathbb{P}(\MF_1|\bx)$ for sample sizes going from 1 to 1000.
Each posterior approximation is based on $10^4$ Metropolis-Hastings iterations.}
\label{norlap2}
\end{figure}

\end{example}

\vspace{0.3cm}
\begin{example}\label{ex:GoF}

In this example we investigate a nonparametric goodness-of-fit problem of testing if the data come from a Gaussian distribution or not. We represent non Gaussian distributions as nonparametric mixtures of Gaussian distributions so that our encompassing model becomes (with an abuse of notations) 
$$ \model_\alpha : \, \alpha  \mathcal N(\mu_1, \sigma_1^2) + (1- \alpha) \int_{\mathbb R} \mathcal N(\mu, \sigma_1^2 ) dP(\mu) $$ 
where we consider a prior distribution on $(\mu_1, \sigma_1^2,  P ) $ defined by 
 $$ \mu_1| \sigma_0^2  \sim \mathcal N(0, \tau^2 \sigma_1^2 ) , \quad \sigma_1^2 \sim IG(b_1,b_2) , \quad P \sim DP ( M, \mathcal N(0, \sigma_1^2 \tau^2)) $$ 
 where $DP ( M, G)$ denotes the Dirichlet process with base measure $M G$ and $IG(b_1,b_2)$ the inverse Gamma distribution with parameter $(b_1,b_2) $. This model defines a standard nonparametric prior distribution on the density $f$ of the observations. 
 
Although the model does not follow the theory developed in Section
\ref{sec:consix} since it is restricted to the parametric case, the general
theory on nonparametric mixture models implies that  the posterior distribution
on $f$ concentrates under $\model_\alpha$ 
around the true density in Hellinger or $L_1$;   see, for instance, \cite{kruijer:rousseau:vdv:10} or \cite{ghosal:vdv:mix}. This implies that if the true distribution with density $f_0$ is not  Gaussian, i.e. 
 $$ \inf_{\mu, \sigma} \| f_0  - \varphi_{\mu, \sigma}\|_1 = \delta >0,$$ where  $ \varphi_{\mu, \sigma}$ is the density of a $ \mathcal N(\mu, \sigma^2)$ random variable, then the posterior probability  
$ \Pi( \alpha > 1 - \delta/2-\epsilon   | \mathbf x^n ) $ for all $\epsilon>0$, goes to 0 almost surely under $f_0$. This is a consequence of 
 $$ \| f_0  -\alpha  \varphi_{\mu_1, \sigma_1} + (1- \alpha) \int_{\mathbb R}\varphi_{\mu, \sigma_1}dP(\mu)  \|_1\geq 
\| f_0  -\alpha  \varphi_{\mu_1, \sigma_1}\|_1 - (1-\alpha) \geq \| f_0  - \varphi_{\mu_1, \sigma_1}\|_1-2 (1-\alpha).$$
The convergence under  $f_0 =  \varphi_{\mu_0, \sigma_0}$ is more intricate but the following heuristic argument gives us some hints on how to choose the hyperparameters: using \cite{scricciolo2011}  we find that the posterior distribution concentrates around $f_0 $ at the rate $\sqrt{\log n}/\sqrt{n}$.  In \cite{long}, it is proved that for nonparametric location mixtures the posterior distribution on the mixing density is Wasserstein  consistent. Here the model is a location mixture of Gaussians, but the common scale is also unknown, and we conjecture the result of  \cite{long} still holds in our case. Hence assuming that the posterior distribution of $Q_\alpha = (\alpha \delta_{(\mu_1)} + ( 1 - \alpha) P)\times\delta_{(\sigma_1)}$ converges in $L_2$-Wasserstein distance to $\delta_{(\mu_0)}\times\delta_{(\sigma_0)}$, we consider  a Taylor expansion and we obtain (see \cite{long} and \cite{rousseau:mengersen:2011}) 
 \begin{equation*}
 \begin{split}
(\log n/n)^{1/2} &\gtrsim   \|  \varphi_{\mu_0, \sigma_0}  -\alpha  \varphi_{\mu_1, \sigma_1} + (1- \alpha) \int_{\mathbb R}\varphi_{\mu, \sigma_1}dP(\mu)  \|_1\\
   & = \left\|   \frac{1}{2} \left(L^{"}_\mu [ E_{Q_\alpha} (\mu- \mu_0)^2  + 2 \sigma_0 (\sigma_1 - \sigma_0)]  + (\sigma_1 - \sigma_0)^2 L_{\sigma,\sigma}^{"} + 2 (\bar \mu - \mu_0)  (\sigma_1 - \sigma_0)L_{\sigma,\mu}^{"} \right) \right. \\
   & \quad \quad +\left. (\bar \mu - \mu_0)\nabla_\mu \varphi + o( u_n ) \right\|_1
 \end{split}
 \end{equation*}
 where  $ \bar \mu = E_{Q_\alpha}(\mu)  $ , $u_n = | \bar \mu - \mu_1| + | E_{Q_\alpha} (\mu- \mu_0)^2+ 2 \sigma_0 (\sigma_1- \sigma_0) | +  (\sigma_1 - \sigma_0)^2$ and  $ L^{"}_\mu$,  $L_{\sigma,\mu}^{"} $ and $L_{\sigma,\sigma}^{"} $ ) are the second derivative of $\varphi_{\mu,\sigma}$ with respect to $\mu$,  $(\mu, \sigma)$ and $\sigma$ respectively. By linear independence, this leads to 
 $|u_n| \lesssim \sqrt{\log n/n} $. In particular, the prior mass of this event if $1 - \alpha < \epsilon $ is bounded by a term of order  
 $ (\log n/n)^{ 1 + M_0(1-e)/4 + (a_0\wedge 1/2)/4}$ for any $1>e >0$ , which is $ o(n^{-1-a_0/2} ) $ as soon as $M_0   + a_0\wedge 1/2> 2a_0 $. Hence, using the same argument as in \cite{rousseau:mengersen:2011} under the Gaussian model the posterior distribution on $\alpha$ will concentrate around 1.
 
This reasoning leads us to consider hyperparameters satisfying $M_0   +
a_0\wedge 1/2> 2a_0 $, for instance: $ a_0 = 1 , M_0 >3/2$ or $a_0 = 1/2$ and
$M_0 =1$.  We implemented an MCMC algorithm using a marginal representation for
the mixture, that is, integrating out the parameters $\mu_1, \sigma_1^2$ and
$P$ and sampling purely $\alpha$ and the allocation random variables in the
data augmentation scheme. The output of this implementation is illustrated in
Figure \ref{Gofalf} for both Normal and non-Normal ($t$ distributed) samples, showing a departure
away from $\alpha=1$ for the later, the slower the decrease the larger the
degree of freedom.
 
\begin{figure}[!h]
\includegraphics[width=.59\textwidth]{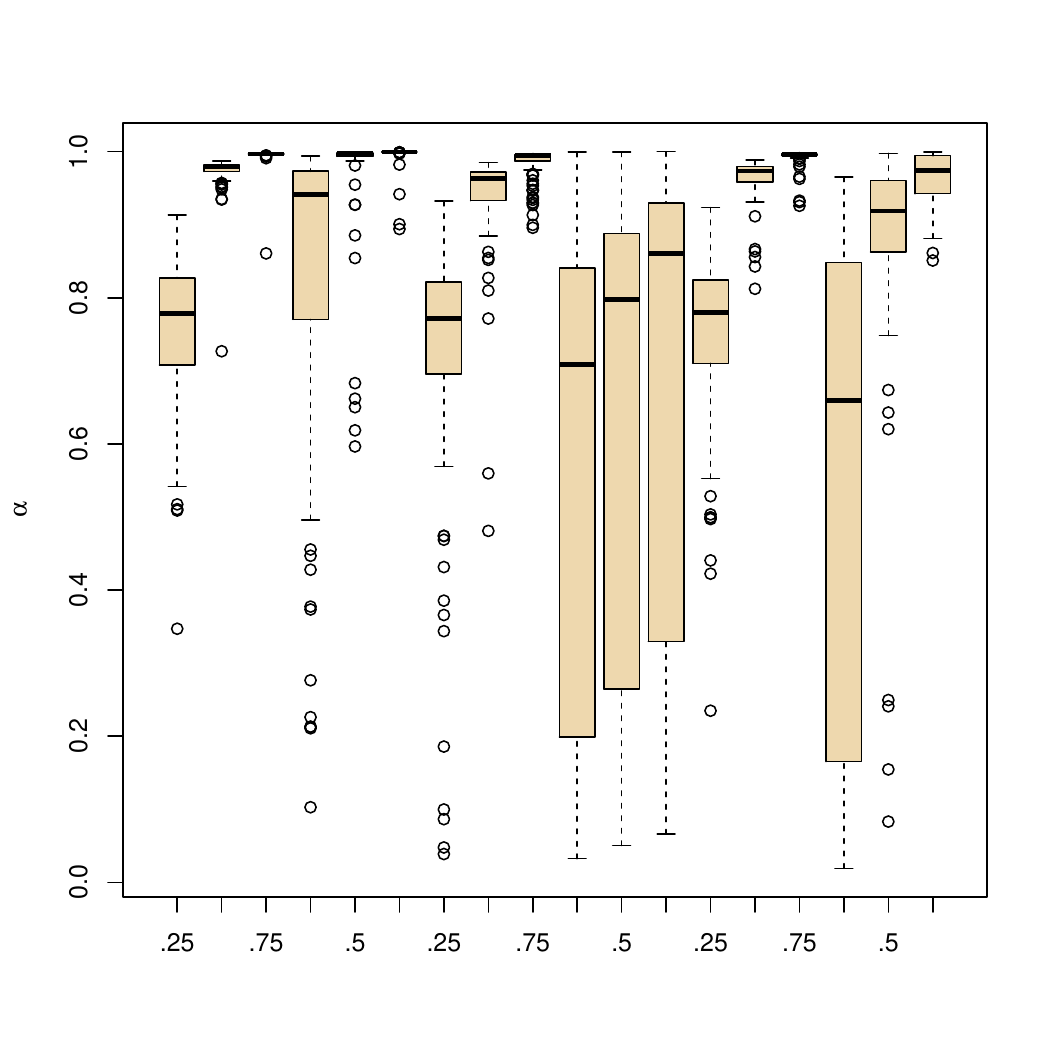}
\caption{\small {\bf Example \ref{ex:GoF}:} Boxplot of the posterior $25\%$, $50\%$, and $75\%$ quantiles of the mixture weight
$\alpha$ of the Normal component for 100 replications of $N$ simulations from (a) a standard normal distribution $(N=100)$;
(b) a standard normal distribution $(N=1000)$; (c) a $t_6$ distribution $(N=100)$; (d) a $t_6$ distribution $(N=100)$; (e) a $t_9$ distribution $(N=100)$; (g) a $t_9$ distribution $(N=100)$. All values are based on $2\,10^4$ Metropolis-Hastings iterations and $100$ replications of the MCMC runs.}
\label{Gofalf}
\end{figure}
\end{example}

\vspace{0.3cm}
\begin{example}\label{ex:Pima}
In this last example we demonstrate that the theory and methodology corresponding to Theorem \ref{cons:para} can be extended to 
the regression case under the assumption that the design is random. 
We consider a binary response setup, using the R dataset about diabetes in Pima Indian
women \citep{cran} as a benchmark \citep[as in][]{marin:robert:2007}. The dataset contains a random sample of 200
women tested for diabetes according to WHO criteria. The response variable $y$ is ``Yes'' or
``No'', for presence or absence of diabetes and the explanatory variable ${\mathbf x}$ is restricted here to the body mass index (bmi) 
$\text{weight in kg}/(\text{height in m})^2$. For this problem, either logistic or probit regression models 
could be suitable, so we compare these fits via our method.
If ${\mathbf y}=(y_1\quad y_2 \ldots y_n)$ is the vector of binary responses and $X=[I_n\quad {\mathbf x}_1]$ is  the 
$n \times 2$ matrix of corresponding explanatory variables, the models in competition can be defined as ($i=1,\ldots,n$)
\begin{align}\label{eq:lp1}
\MF_1: y_i \mid {\mathbf x}^i, \theta_1 &\sim \mathcal{B}(1, p_i) \quad \text{where} \quad p_i=  \frac{\exp({\mathbf x}^i \theta_1)}{1+\exp({\mathbf x}^i \theta_1)} \nonumber \\
\MF_2: y_i \mid {\mathbf x}^i, \theta_2 &\sim \mathcal{B}(1, q_i) \quad \text{where} \quad q_i=  \Phi({\mathbf x}^i \theta_2) 
\end{align}
where ${\mathbf x}^i=(1\ x_{i1})$ is the vector of explanatory variables and where  $\theta_j$, $j=1, 2$, is a
$2 \times 1$ vector made of the intercept and of the regression coefficient under either $\MF_1$ or $\MF_2$. We once
again consider the case where both models share the same parameter. However, for this generalised linear model there 
is no moment equation that relates $\theta_1$ and $\theta_2$, so we adopt a local reparameterisation strategy by rescaling 
the parameters of the probit model $\MF_2$ so that the MLE's of both models coincide. This strategy follows from the remark by
\cite{choudhury:ray:sarkar:2007} regarding the connection between the Normal cdf and a logistic function
$$
  \Phi({\mathbf x}^i \theta_2) \approx \frac{\exp(k {\mathbf x}^i \theta_2)}{1+\exp(k {\mathbf x}^i \theta_2)}
$$
and we attempt to find the best estimate of $k$ to make both parameters coherent. Given
$$
(k_0, k_1)=(\nicefrac{\widehat{\theta_{01}}}{\widehat{\theta_{02}}},
\nicefrac{\widehat{\theta_{11}}}{\widehat{\theta_{12}}})\,,
$$
which denote ratios of the maximum likelihood estimates of the logistic model parameters to those for the probit model,
we redefine $q_i$ in \eqref{eq:lp1} as
\begin{equation}\label{eq:lp2}
q_i=  \Phi({\mathbf x}^i (\kappa^{-1}\theta))\,, 
\end{equation}
$\kappa^{-1}\theta=(\nicefrac{\theta_0}{k_0},\nicefrac{\theta_1}{k_1})$.
 
Once the mixture model is thus parameterised, we set our now standard Beta $\mathcal{B}(a_0, a_0)$ on the weight of
$\MF_1$, $\alpha$, and choose the default $g$-prior on the regression parameter \citep[see, e.g.,][Chapter
4]{marin:robert:2007}, so that
$$
\theta \sim \mathcal{N}_{2}(0, n(X^{\text{T}}X)^{-1}) .
$$
In a Gibbs representation (not implemented here), the full conditional posterior distributions given the allocation
vector $\zeta$ are $\alpha \sim \mathcal{B}(a_0+n_1, a_0+n_2)$ and
\begin{equation}\label{eq:lp3}
\begin{split}
\pi (\theta \mid {\mathbf y}, X, \zeta) &\propto \dfrac{\exp\left\{ \sum_i \mathbb{I}_{\zeta_i=1}y_i{\mathbf x}^i
\theta \right\}}{\prod_{i; \zeta_i=1}[1+\exp({\mathbf x}^i \theta)]} \,\exp\left\{ -\theta^T(X^TX)\theta\big/ 2n \right\} \\
& \times 
\prod_{i; \zeta_i=2} \Phi({\mathbf x}^i (\kappa^{-1}\theta))^{y_i}(1-\Phi({\mathbf x}^i (\kappa^{-1}\theta)))^{(1-y_i)}
\end{split}
\end{equation}
where $n_1$ and $n_2$ are the number of observations allocated to the logistic and probit models, respectively. This
conditional representation shows that the posterior distribution is then clearly defined, which is obvious when
considering that the chosen prior is proper.
 
For the Pima dataset, the maximum likelihood estimates of the GLMs
are $\hat{\theta_1}=(-4.11, 0.10)$ and $\hat{\theta_2}=(-2.54, 0.065)$, respectively, and so $k=(1.616, 1.617)$.
We compare the outcomes of this Bayesian analysis when $a_0=.1, .2, .3, .4, .5$ in Table \ref{lpt}. As clearly shown in 
the Table, the estimates of $\alpha$ are close to $0.5$ for all values of $a_0$ and the estimates of
$\theta_0$ and $\theta_1$ are very stable (and quite similar to the MLEs). We note a slight increase of $\alpha$ towards
$0.5$ as $a_0$ increases, but do not want to over-interpret the phenomenon. This behaviour leads us to conclude that
(a) none or both of the models are appropriate for the Pima Indian data, and (b) the sample size may be insufficiently large
to allow discrimination between the logit and the probit models. 

\begin{table}
  \caption{\small \label{lpt} Dataset Pima.tr: Posterior medians of the mixture model parameters.}
\begin{tabular}{@{}llllll@{}}
            &                  &\multicolumn{2}{l}{\bf Logistic model parameters}&\multicolumn{2}{l}{\bf Probit model parameters}\\
$a_0$& $\alpha$ & $\theta_{0}$ & $\theta_{1}$ & $\frac{\theta_0}{k_0}$ & $\frac{\theta_1}{k_1}$ \\

.1        & .352 & -4.06 & .103 & -2.51 & .064 \\
.2        & .427 & -4.03 & .103 & -2.49 & .064 \\
.3        & .440 & -4.02 & .102  & -2.49 & .063  \\ 
.4        & .456   & -4.01 & .102 & -2.48  &  .063\\
.5        & .449& -4.05 & .103 & -2.51 & .064 \\
\end{tabular}
\end{table}

To follow up on this last remark, we ran a second experiment with
simulated logit and probit datasets and a larger sample size $n=10,000$. We used the regression
coefficients $(5, 1.5)$ for the logit model and $(3.5, .8)$ for the probit model. The estimates of the
parameters of both $\MF_{\alpha_1}$ and $ \MF_{\alpha_2}$ and for both datasets are presented in Table \ref{lpt1}. For
every $a_0$, the estimates in the true model are quite close to the true values and the posterior estimates of $\alpha$
are close to $1$ in the logit case and to $0$ in the probit case. For this large setting, there is thus
consistency in the selection of the proper model. In addition, Figure \ref{lpf} shows that when the sample size is large
enough, the posterior distribution of $\alpha$ concentrates its mass near $1$ and $0$ when the data are simulated from a
logit and a probit model, respectively.

\begin{table}
  \caption{\small \label{lpt1} Simulated dataset: Posterior medians of the mixture model parameters.}
\begin{tabular}{@{}lllllllllll@{}}
 & & & \bf $\MF_{\alpha}^1$ & & & & & \bf $\MF_{\alpha}^2$ & & \\
True model: & \multicolumn{5}{l}{\bf logistic with $\theta_1=(5, 1.5)$}&\multicolumn{5}{l}{\bf probit with $\theta_2=(3.5, .8)$}\\
$a_0$& $\alpha$ & $\theta_{0}$ & $\theta_{1}$ & $\frac{\theta_0}{k_0}$ & $\frac{\theta_1}{k_1}$ & $\alpha$ & $\theta_{0}$ & $\theta_{1}$ & $\frac{\theta_0}{k_0}$ & $\frac{\theta_1}{k_1}$\\

.1        & .998 & 4.940 & 1.480 & 2.460 & .640 & {}.003& 7.617& 1.777& {}3.547& {}.786 \\
.2        & .972 & 4.935 & 1.490 & 2.459 & .650 & {}.039& 7.606& 1.778& {}3.542& {}.787 \\
.3        & .918 & 4.942 & 1.484 & 2.463 & .646 & {}.088& 7.624& 1.781& {}3.550& {}.788 \\ 
.4        & .872 & 4.945 & 1.485 & 2.464 & .646 & {}.141& 7.616& 1.791& {}3.547& {}.792  \\
.5        & .836 & 4.947 & 1.489 & 2.465 & .648 & {}.186& 7.596& 1.782& {}3.537& {}.788 \\
\end{tabular}
\end{table}

\begin{figure}
\centering
\begin{subfigure}[b]{0.32\textwidth}
\includegraphics[width=\textwidth]{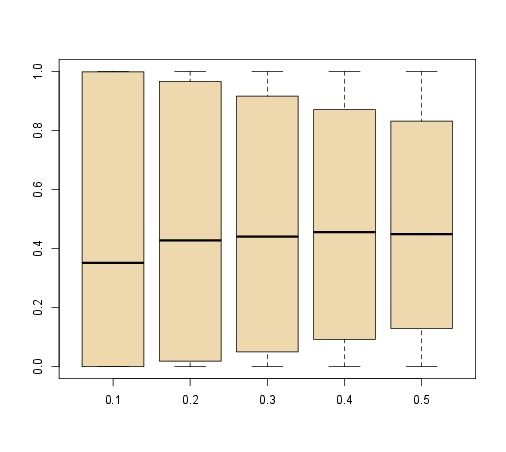}
\caption{Pima dataset}
                \label{fig:gulle}
        \end{subfigure}
 ~ 
        \begin{subfigure}[b]{0.32\textwidth}
\includegraphics[width=\textwidth]{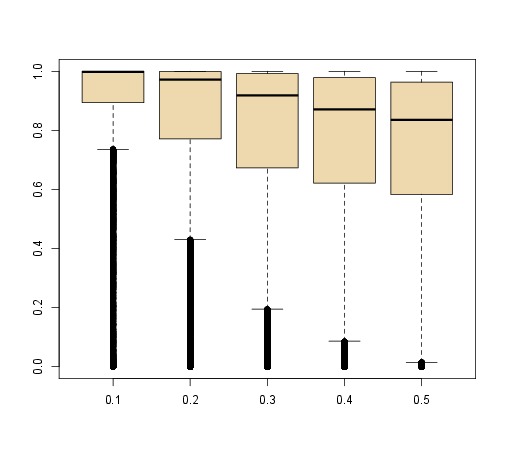}
\caption{Data from logistic model}
                \label{fig:tigera}
        \end{subfigure}
 ~ 
        \begin{subfigure}[b]{0.32\textwidth}
\includegraphics[width=\textwidth]{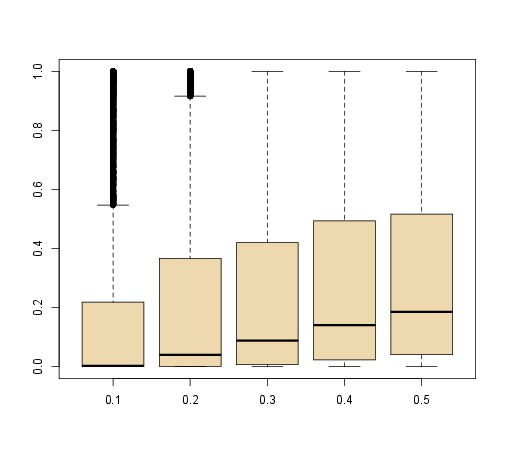}
 \caption{Data from probit model}
                \label{fig:mouseb}
        \end{subfigure}
\caption{\small {\bf Example \ref{ex:Pima}:} Histograms of the posterior distributions of $\alpha$ in favor of the logistic model based on $10^4$ Metropolis-Hastings iterations where $a_0=.1, .2, .3, .4, .5$.}
\label{lpf}
\end{figure}
\end{example}

\vspace{0.3cm}
\section{Conclusion}\label{sec:quatr}

Bayesian inference has been used in a very wide and increasing
range of contexts over the past thirty years, and many of the applications 
of the Bayesian paradigm have concentrated on comparing scientific theories and testing hypotheses. 
Due to the ever increasing complexity of the statistical models handled in
such applications, the natural and understandable tendency of practitioners has been to rely on the default solution of
the posterior probability (or equivalently of the Bayes factor) without fully understanding the sensitivity of these methods
to both prior modeling and posterior calibration \citep{robert:cornuet:marin:pillai:2011}. In this area, 
objective Bayes solutions remain tentative and have not reached consensus.

The novel approach we have proposed here for Bayesian testing of hypotheses and Bayesian model comparison offers in our opinion 
many incentives over these established methods. By casting the problem as an encompassing mixture model, not only do 
we replace the original testing problem with a better controlled estimation target that focuses on the frequency of a given 
model within the mixture model, but we also allow for posterior variability of this frequency. The posterior distribution of 
the weights of both components in the mixture offers a setting for deciding about which model is most favored by the 
data that is at least as intuitive as the sole number corresponding to either the posterior probability or the Bayes factor. 
The range of acceptance, rejection and indecision conclusions can easily be calibrated by simulation under both models, 
as well as by deciding on the values of the weights that are extreme enough in favor of one model. The examples provided 
in this paper have shown that the posterior medians of such weights settle very quickly near the boundary values $1$. 
Although we do not advocate such practice, it is even possible to derive a Bayesian $p$-value by considering the posterior 
area under the tail of the distribution of the weight. Moreover, the approach does not induce additional computational strain
on the analysis.

Besides decision making, another issue of potential concern about this new approach is the impact of the prior
modelling. As demonstrated in our examples, a partly common parameterisation is often feasible and hence allows
for reference priors, at least on the common parameters. This proposal thus allows for a partial removal of the prohibition 
on using improper priors in hypothesis testing \citep{degroot:1973}, a problem which has plagued the
objective Bayes literature for decades. Concerning the prior on the weight parameter, we analyzed the sensitivity on the
resulting posterior distribution of various prior Beta modelings on those weights. While the sensitivity is clearly
present, it naturally vanishes as the sample size increases, in agreement with our consistency results, and remains of a
moderate magnitude. This leads us to suggest the default value of $a_0=0.5$ in the Beta prior, in connection with both the
earlier result of \cite{rousseau:mengersen:2011} and Jeffreys' prior in the simplest mixture setting.

\hyphenation{Post-Script Sprin-ger}

\small
\section*{Appendix 1: Proofs of Section 3}\label{app:proofs} 
In this Section we give the proofs of Theorems \ref{cons:para}, \ref{th:cons:embed} and \ref{th:separationrate}. 

\subsection*{Proof of Theorem \ref{cons:para}} 

Using Proposition \ref{prop:cons}, we have that 
$$
\pi\left( A_n | \mathbf x^n \right) = 1+ o_p(1)
$$
with $A_n = \{ ( \alpha, \theta); \|f_{\theta, \alpha}-f_{\theta^*, \alpha^*} \|_1 \leq \delta_n\}$ and $\delta_n = M
\sqrt{ \log n/n}$. Consider a subsequence $\alpha_n, P_{1,\theta_{1n}}, P_{2,\theta_{2n}}$ which converges to $\alpha,
\mu_1, \mu_2$ where convergence holds in the sense that $\alpha_n \rightarrow \alpha$ and $P_{j,\theta_{jn}}$ converges
weakly to $\mu_j$.  Note that $\mu_j(\mathcal X)\leq 1$ by precompacity of the unit ball under the  weak topology. At
the limit $$ \alpha \mu_1 + (1-\alpha)\mu_2 = \alpha^* P_{1, \theta_1^*}+ ( 1 - \alpha^*)P_{2, \theta_2^*}$$
The above equality implies that $\mu_1$ and $\mu_2$ are probabilities. Using \eqref{ident}, we obtain that 
$$ 
\alpha = \alpha^*, \quad \mu_j = P_{j, \theta_j^*},
$$
which implies posterior consistency for $\alpha$. The proof of \eqref{rate} follows the same line as in
\cite{rousseau:mengersen:2011}. Consider  first the case where $\alpha^* \in (0,1)$. Then the posterior distribution on
$\theta$ concentrates around $\theta^* $. 

Writing 
$$
L' = (f_{1,\theta_1^*}-f_{2,\theta_2^*},\alpha^* \nabla f_{1,\theta_1^*}, (1- \alpha^*)\nabla f_{2,\theta_2^*}  ) := (L_\alpha, L_1, L_2)$$
$$
L^{"} = \mbox{diag}(0, \alpha^*  D^2 f_{1,\theta_1^*}, (1- \alpha^*)D^2 f_{2,\theta_2^*}  )\quad \mbox{and } \quad \eta
= ( \alpha - \alpha^*, \theta_1-\theta_1^*, \theta_2-\theta_2^*), \quad \omega = \eta /| \eta|\,,
$$
we then have
\begin{equation} \label{taylor}
 \|  f_{ \theta, \alpha} - f_{ \theta^*,\alpha^*}\|_1 = |\eta|	\left| \omega^{\text{T}} L' + \nicefrac{ |\eta| }{ 2 } \omega^{\text{T}} L^{"}
\omega + |\eta| \omega_1\left[ \omega_2^T L_1- \omega_3^T L_2 \right]+ o(|\eta|) \right|
\end{equation}
For all $(\alpha, \theta) \in A_n$, $\eta = (\alpha -\alpha^*, \theta_1-\theta_1^*, \theta_2 - \theta_2^*)$ goes to
0 and for $n$ large enough there exists $\epsilon>0$ such that  $|\alpha-\alpha^*| + |\theta-\theta^*|\leq \epsilon$. We
now prove that there exists $c>0$ such that for all  $(\alpha, \theta) \in A_n$
$$
v(\omega) = \left| \omega^{\text{T}} L' + \frac{ |\eta| }{ 2 } \omega^{\text{T}} L^{"} \omega + |\eta| \omega_1\left[ \omega_2^{\text{T}} L^{'}_2+
\omega_3^{\text{T}} L^{'}_3 \right]+ o(|\eta|) \right|>c,
$$
where $\omega$ is defined with respect to $\alpha, \theta$.
Were it not the case, there would exist a sequence $(\alpha_n,\theta_n) \in A_n $ such that the associated $v(\omega_n)\leq
c_n$ with $c_n=o(1)$. As $\omega_n$ belongs to a compact set we could find a subsequence converging to a point $\bar
\omega$. At the limit we  would obtain 
$$
\bar \omega ^{\text{T}} L^{'}= 0
$$
and by linear independence $\bar \omega=0$ which is not possible.  Thus for all $(\alpha, \theta)\in A_n$ 
$$
|\alpha-\alpha^*| + |\theta-\theta^*| \lesssim  \delta_n. 
$$
Assume now instead that $\alpha^*= 0$. Then define $ L^{'} = (L_\alpha, L_2)$ and
$$
L^{"} = \mbox{diag}(0,  D^2 f_{2,\theta_2^*}  )\quad \mbox{and } \quad \eta = ( \alpha - \alpha^*, \theta_2-\theta_2^*),
\quad \omega = \eta /| \eta|
$$
and consider a Taylor expansion with $\theta_1$ fixed, $\theta_1^* = \theta_1$ and $|\eta|$ going to 0. This leads to 
\begin{equation} \label{taylorbis}
\|  f_{ \theta, \alpha} - f_{\alpha^*, \theta^*}\|_1 = 
|\eta|	\left| \omega^{\text{T}} L' + \frac{ |\eta| }{ 2 } \omega^{\text{T}} L^{"} \omega 
- |\eta| \omega_1 \omega_3 L_2 \right| + o(|\eta|)  
\end{equation}
in place of \eqref{taylor} and using the same argument as in the case $\alpha^*\in (0,1)$, $$
|\alpha-\alpha^*| + |\theta-\theta^*| \lesssim  \delta_n. 
$$

\subsection*{Proof of Theorem \ref{th:cons:embed}}

Recall that $f^* = f_{1, \theta_1^*}$. To prove \eqref{rate:sharp}
we must first find a precise lower bound on 
$$
D_n:=  \int_\alpha \int_\Theta e^{l_n(f_{\theta, \alpha})- l_n(f^*)}\text{d}\pi_{\theta}(\theta)\text{d}\pi_{\alpha}(\alpha)
$$
Consider the approximating set
$$ 
S_n(\epsilon) = \{ (\theta, \alpha), \alpha > 1-1/\sqrt{n}, |\theta_1-\theta_1^*| \leq 1/\sqrt{n}, |\psi - \bar \psi| \leq \epsilon\}, \quad \theta = (\theta_1, \theta_2)
$$
with $|\bar \psi | > 2 \epsilon$ some fixed parameter in $\mathcal S$. Using the same computations as in \cite{rousseau:mengersen:2011}, 
it holds that  for all $\delta >0$ there exists $C_\delta >0$ such that 
\begin{equation}\label{denom}
    P^*\left( D_n < e^{-C_\delta}\pi(S_n(\epsilon))/2 \right) < \delta.
\end{equation}
So that with probability greater than $1-\delta$, $D_n \gtrsim n^{- (a_2 + d_1)/2}$.
Denote $B_n = \{ (\theta, \alpha); \| f_{\theta,\alpha} - f^*\|_1\leq M_n / \sqrt{n}\}$, we know that 
 $$\pi( \| f_{\theta,\alpha} - f^*\|_1\leq M_n\sqrt{\log n} / \sqrt{n} | \mathbf x^n) =o_p(1).$$
Let $M_n \leq j \leq M_n\sqrt{\log n}$ and consider the slice 
$ S_n(j) = \{ j /\sqrt{n} \leq \| f_{\theta,\alpha} - f^*\|_1\leq (j+1)/ \sqrt{n}\}$. 
We now upper bound $\pi(S_n(j))$. To do so we split the parameter space into $ \alpha < 1 -\delta$ for a fixed arbitrarily small $\delta$,  and $\alpha > 1 -\delta$. In the first case, we have that $\psi$ converges to 0 and $\theta_1$ to $\theta_1^*$. A Taylor expansion leads to 
\begin{equation} \label{case2a}
\begin{split}
& \| \alpha f_{1,\theta_1}+ (1-\alpha) f_{2, \theta_1, \psi} - f_{1, \theta_{1}^*}\|_1  \\
 &= \|(\theta_1 - \theta_{1}^* )^T \nabla_\theta f_{1,\theta_{1}^*} + (1 - \alpha )\psi^T \nabla_\psi f_{2,\theta_{1}^*,0}  \|_1+O(\|\theta_1 - \theta_{1}^*\|^2 + (1-\alpha)\|\psi\|^2 )
\end{split}
\end{equation} 
Setting $v_n = \|\theta_1 - \theta_{1}^* \| + \|(1 - \alpha )\psi \|$ and $\eta = (\theta_1 - \theta_{1}^*, (1-\alpha)\psi)/v_n$, \eqref{case2a} implies that
$$\| \alpha f_{1,\theta_1}+ (1-\alpha) f_{2, \theta_1, \psi} - f_{1,\theta_{1}^*}\|_1  \geq v_n\|\eta^T \nabla f_{2,\theta_{1}^*,0}   \|_1 + O(v_n^2)$$
and by  linear independence of $\nabla f_{2,\theta_1, \psi} $ we obtain   that $v_n \leq C j/\sqrt{n}$ on $S_n(j) \cap \{ \alpha \in (0, 1-\delta) \} $ and 
 $$\pi_{n,1} := \pi( S_n(j) \cap \{ \alpha \in (0, 1-\delta) \}) \lesssim j^{d_2} n^{-d_2/2}. $$
Now consider $1-\alpha \leq \delta $. If $1-\alpha \leq M j /\sqrt{n}$ then  $\| \theta_1 -\theta_{1}^*\| \lesssim j/\sqrt{n}$ which has prior probability bounded by 
 $O( (j/\sqrt{n})^{d_1+a_2})$. If $1-\alpha > M j /\sqrt{n}$, then $\psi $ goes to 0 and \eqref{case2a} implies that 
 $$ \|\theta_1 - \theta_{1}^* \| + \|(1 - \alpha )\psi \| \lesssim j/\sqrt{n} $$
which in turns implies that 
$$ \pi(S_n(j)) \lesssim  j^{d_2} n^{-d_2/2} + (j/\sqrt{n})^{d_1+a_2}+ j^{d_2} n^{-d_2/2}\int_{Mj/\sqrt{n}}^\delta u^{a_2-d_\psi-1}du \lesssim  j^{d_2} n^{-d_2/2} + (j/\sqrt{n})^{d_1+a_2}$$
Since $a_2 \leq d_\psi$, 
\begin{equation}\label{piSjDn}
\frac{\pi(S_n(j))}{\pi(S_n(\epsilon))} \lesssim j^{d_\psi + a_2}.
\end{equation}
Write $S_\epsilon = \{ \epsilon \leq \|f_{ \theta,\alpha} - f^*\|_1 \leq 2\epsilon\}$
Equation \eqref{case2a} implies that if $M_n \sqrt{\log n}/\sqrt{n} \geq \epsilon \geq M_n /\sqrt{n}$, $S_\epsilon \subset \{ \|\theta_1 - \theta_{1}^* \| \leq \tau_1\epsilon \} \cap  \{ \|(1 - \alpha )\psi \| \leq \tau_1 \epsilon\}$ for some $\tau_1 >0$. To cover  $S_\epsilon $ by $L_1$ balls of radius $\epsilon/2$ we consider $\|\theta_1 - \theta_1' \| \leq  \zeta \epsilon$ and we split the set into $1-\alpha \leq \epsilon$ and $ 1-\alpha > \epsilon$. 
Note that by choosing $\zeta$ small enough, 
\begin{equation*}
\begin{split}
 \|f_{\theta, \alpha} - f_{ \theta',\alpha'}\|_1 &\leq  \|f_{\theta_1,\psi, \alpha} - f_{\theta_1,\psi', \alpha'}\|_1 + \epsilon/16\\
 & = \|(1-\alpha) (f_{2,\theta_1,\psi} - f_{2,\theta_1,0}) - (1-\alpha') (f_{2,\theta_1,\psi'} - f_{2,\theta_1,0}) \|_1 + \epsilon/16 
\end{split} 
\end{equation*}
If $1-\alpha > \kappa \epsilon$ for some $\kappa >0$. Then on $S_\epsilon $ , $\| \psi\| \leq \tau_1 /\kappa$ and by choosing $\kappa$ large enough, 
\begin{equation*}
\begin{split}
 \|f_{  \theta,\alpha} - f_{ \theta', \alpha'}\|_1 &\leq   M_1 \| (1-\alpha) \psi  - (1-\alpha') \psi' \| + \epsilon/8 \leq \epsilon/4
\end{split} 
\end{equation*}
 by choosing $\| (1-\alpha) \psi  - (1-\alpha') \psi' \| \leq \zeta \epsilon$ with $\zeta $ small enough. 
If $1-\alpha \leq \kappa \epsilon$, choose $|\alpha' - \alpha| \leq  \epsilon/16$ so that 
\begin{equation*}
\begin{split}
 \|f_{  \theta,\alpha} - f_{ \theta', \alpha'}\|_1  
 & \leq (1-\alpha)\|f_{2,\theta_1,\psi} - f_{2,\theta_1,\psi'} \|_1 + \epsilon/4 \leq  (1-\alpha)\|\psi - \psi' \| M_1  + \epsilon/8 \leq \epsilon/2 
 \end{split} 
\end{equation*}
by choosing $\|\psi - \psi' \| \leq 1/(4M_1\kappa)$. Hence the local $L_1$ entropy is bounded by a constant for all $M_n \sqrt{\log n}/\sqrt{n} >\epsilon >M_n /\sqrt{n}$ and using Theorem 2.4 of \citep{ghosal:ghosh:vdv:00}, we obtain \eqref{rate:sharp}. 

Now consider $d_\psi > a_2$ and let $A_n = \{ (\theta, \alpha) \in B_n; 1-\alpha > z_n/\sqrt{n} \}$ with $z_n$  a sequence
increasing to infinity faster than $M_n$ and $B_n = \{   \|f_{  \theta,\alpha} - f^*\|_1 \leq M_n /\sqrt{n} \} $. We prove that $\pi( A_n| \mathbf{x}^n) = o_p(1) $ by proving that $\pi( A_n) =o( n^{- (a_2 + d_1)/2})$ and using the lower bound on $D_n$ of order $n^{- (a_2 + d_1)/2}$. We split $B_n$ into 
\begin{equation*}
B_{n,1}(\delta)  = B_n \cap \{ (\theta, \alpha), \theta = (\theta_1, \psi); \|\psi\| < \delta \}, \quad
B_{n,2}(\delta) = B_n \cap B_{n,1}(\delta)^c , \quad \delta >0
\end{equation*}
To simplify notation we also write $\delta_n = M_n/\sqrt{n}$. First we prove that for all $\delta>0$,
$A_n \cap   B_{n,2}(\delta)  = \emptyset$, when $n$ is large enough. Let $\delta>0$, then for any $(\theta, \alpha)
\in A_n \cap B_{n,2}(\delta)$, We thus  have $\|\psi\| \neq o(1)$, $\alpha = 1 + o(1) $ and $|\theta_1 - \theta_1^*|= o(1)$.
Consider a Taylor expansion of $f_{\theta,\alpha}$ around $\alpha = 1$ and $\theta_1 = \theta_1^*$ , with $\psi$ fixed.
This leads to 
\begin{equation*}
\begin{split}
f_{\theta,\alpha}- f^* &= (\alpha-1) [ f_{1,\theta_1^*} - f_{2, \theta_1^*, \psi}]
 + ( \theta_1- \theta_1^*)
[\nabla_{\theta_1} f_{1,\theta_1^*}  -  \nabla_{\theta_1} f_{2, \theta_1^*,\psi }(x)] \\
&  \quad + 
\frac{1}{2}( \theta_1- \theta_1^*)^{\text{T}} \left( \bar \alpha  D^2_{\theta_1} f_{1,\bar \theta_1} + (1- \bar \alpha)D^2_{\theta_1} f_{2, \bar
\theta_1, \psi}\right)( \theta_1- \theta_1^*) \\
& \quad + (\alpha-1)( \theta_1- \theta_1^*)^{\text{T}} [ \nabla_{\theta_1}f_{1,\bar\theta_1} - \nabla_{\theta_1}f_{2,
\bar\theta_1, \psi}] \\
&=  (\alpha-1) [ f_{1,\theta_1^*} - f_{2, \theta_1^*, \psi}] + ( \theta_1- \theta_1^*) [\nabla_{\theta_1}
f_{1,\theta_1^*}   -  \nabla_{\theta_1} f_{2, \theta_1^*,\psi }(x)] + o(|\alpha-1|+\|\theta_1 - \theta_1^*\|) \\  
\end{split}
\end{equation*}
with $\bar \alpha \in (0,1)$ and $\bar \theta_1 \in (\theta_1, \theta_1^*)$ and the $o(1)$ is uniform over $A_n \cap B_{n,2}(\delta)$.
Set $\eta = ( \alpha-1, \theta_1 - \theta_1^*)$ and $x = \eta / |\eta|$ if $|\eta |>0$. Then 
\begin{equation*}
\| f_{\theta,\alpha} - f^* \|_1 = |\eta|\left( x^{\text{T}} L_1(\psi) + o(1)\right), \quad L_1 = ( f_{1,\theta_1^*} -
f_{2, \theta_1^*, \psi}, \nabla_{\theta_1} f_{1,\theta_1^*}  -  \nabla_{\theta_1} f_{2, \theta_1^*,\psi }(x))
\end{equation*}
We now prove that on $A_n \cap B_{n,2}(\delta)$,   $  \| f_{\theta,\alpha} - f^* \|_1 \gtrsim |\eta|$. Assume that it
is not the case; then there exist $c_n>0$ going to 0 and  a sequence $ (\theta_{1,n}, \alpha_n)$ such that along that
subsequence $| x_n^{\text{T}} L_1(\psi_n) +o(1) | \leq c_n$ with $x_n = \eta_n/|\eta_n|$. Since it belongs to a compact set,
together with $\psi_n$, any converging subsequence satisfies at the limit $(\bar x, \bar \psi)$,
$$
\bar x^{\text{T}} L_1(\bar \psi) = 0\,, 
$$
which is not possible. Hence $|\alpha-1| \lesssim M_n/\sqrt{n} = o(w_n/\sqrt{n})$, which is
not possible so that $    A_n \cap   B_{n,2}(\delta)  = \emptyset$ when $n$ is large enough. We now bound 
$\pi(A_n \cap   B_{n,1}(\delta))$ for $\delta>0$ small enough but fixed. We consider a Taylor expansion around
$\theta^* = (\theta_1^*, 0)$, leaving $\alpha $ fixed. Note that $ \nabla_{\theta_1} f_{2, \theta^*} =
\nabla_{\theta_1} f_{1, \theta_1^*}$. We have 
\begin{equation*}
f_{\theta,\alpha} - f^* = (\theta_1 - \theta_1^*)^{\text{T}} \nabla_{\theta_1} f_{1, \theta_1^*} +(1-\alpha)
\psi^{\text{T}} \nabla_\psi f_{2, \theta^*} + \frac{ 1}{ 2} (\theta- \theta^*)^{\text{T}} H_{\alpha, \bar \theta} (\theta-
\theta^*)
\end{equation*}
where $H_{\alpha, \bar \theta}$ is the block matrix 
$$
H_{\alpha, \bar \theta} = \left( \begin{array}{cc} \alpha D^2_{\theta_1} f_{1, \bar \theta_1} + (1 - \alpha) D^2_{\theta_1,\theta_1}
f_{2, \bar \theta}  & (1-\alpha) D^2_{\theta_1, \psi} f_{2, \bar \theta}\\
(1-\alpha) D^2_{ \psi,\theta_1} f_{2, \bar \theta} &  (1-\alpha) D^2_{\psi, \psi} f_{2, \bar \theta} \end{array} \right)
$$
Since $H_{\alpha, \bar \theta}$ is bounded in $L_1$ (in the sense that each of its components is bounded as functions in $L_1$),
uniformly in neighbourhoods of $\theta^*$, we have writing $\eta = (\theta_1 - \theta_1^*, (1 -\alpha) \psi)$ and $x =
\eta /|\eta|$, that $|\eta|=o(1)$ on $A_n \cap B_{n,1}(\delta)$ and
\begin{equation*}
\|  f_{\theta, \alpha} - f^*\|_1 \gtrsim |\eta| \left( x^{\text{T}} \nabla f_{2, \theta^*}+ o(1) \right), 
\end{equation*}
if $\epsilon$ is small enough. Using a similar argument to before, this leads to $|\eta| \lesssim \delta_n$ on $A_n \cap
B_{n,1}(\delta)$, so that 
\begin{equation*}
\begin{split}
\pi\left( A_n \cap B_{n,1}(\delta)\right) &\lesssim \delta_n^{d_1} \int_{z_n/\sqrt{n}}^1 (\delta_n/u)^{d_\psi}u^{a_2-1}du
\lesssim \delta_n^{d_1+d_\psi}z_n^{a_2-d_\psi} n^{(d_\psi-a_2)/2}\lesssim n^{-(d_1+a_2)/2}M_n^{d_2}z_n^{a_2-d_\psi}\\
&  = o( n^{-(d_1+a_2)/2})
\end{split}
\end{equation*}
choosing $M_n = o(z_n^{(d_\psi-a_2)/d_2} )$, going to infinity (recall that $d_\psi > a_2$). 

Finally assume that $d_\psi <  a_2$ and denote $C_n = \{ (\theta, \alpha) \in B_n; 1-\alpha < e_n \}$, then the sam arguments imply that  if $\delta >0$ is small enough, $ C_n = C_n \cap B_{n,1}(\delta)$ and 
\begin{equation*}
\begin{split}
\pi\left( C_n \cap B_{n,1}(\delta)\right) &\lesssim \delta_n^{d_1} \int_{0}^{e_n} (\delta_n/u)^{d_\psi}u^{a_2-1}du
\lesssim \delta_n^{d_1+d_\psi}e_n^{a_2-d_\psi}  \lesssim n^{-d_2/2}M_n^{d_\psi}e_n^{a_2-d_\psi} = o(n^{-d_2/2})
\end{split}
\end{equation*}
if $M_n^{d_\psi}= o(e_n^{-(a_2-d_\psi)} )$. Now, working with 
$ S_n' = \{ \|\theta_1 - \theta_1^*\| \leq 1/\sqrt{n} , \, \|\psi\| \leq 1/\sqrt{n}, \alpha \in (\bar \alpha - e_n' , \bar \alpha + e_n')$ with $e_n' $ going to 0 arbitrarily slowly, we have that with probability going to 1, 
$D_n \gtrsim \pi(S_n') \gtrsim n^{-d_2/2} e_n' $ so that by choosing $e_n'$ accordingly,
$ \pi(C_n)   = o(n^{-d_2/2} e_n') $ and Theorem \ref{th:cons:embed} is proved .

\subsection*{Proof of Theorem \ref{th:separationrate}}

The proof of Theorem \ref{th:separationrate} proceeds along the same line as the previous proof. Let $f_n^*  = f_{2, \theta_{1,n}, \psi_n}$ with $\|\psi_n\| = o(1)$; the other case has already been proved in Theorem \ref{cons:para}. 
Recall that $$\pi\left( \| f_{\theta, \alpha}- f_n^*\|_1\geq M_0 \sqrt{\log n}/\sqrt{n} |\mathbf x^n \right) = o_p(1)$$
if $M_0>0$ is large enough. We restrict ourselves to the case where $M_n /\sqrt{n} \leq \|\psi_n\| = o(n^{-1/4})$ with  since  the other case
can be treated more easily.  We prove first that the posterior concentration rate can be sharpened into 
\begin{equation} \label{sharp:rate}
\pi\left( \| f_{\theta, \alpha}- f_n^*\|_1\geq z_n/\sqrt{n} |\mathbf x^n \right) = o_p(1), \quad \mbox{for any } \quad z_n \rightarrow + \infty .
\end{equation} 
To do so we first obtain a sharp lower bound on $D_n$. From the regularity assumptions [B1] and [B2] 
for all $(\alpha, \theta_1, \psi) \in \tilde S_n = \{ \|\theta_1-\theta_{1,n}\|\leq 1/\sqrt{n}, \|(1 - \alpha)\psi- \psi_n\| + \|\psi_n\| \|\psi\| \leq 1/\sqrt{n}\}$ and
a Taylor expansion with $\theta = (\theta_1, \psi)$ around $\theta_{1,n}$ and 0 (both for $\psi$ and $\psi_n$) we bound, 
\begin{equation*}
\begin{split}
KL( f_n^* , f_{\theta, \alpha} ) &\leq \int \frac{ (f_{\theta_{1,n}, \psi_n} - f_{\theta, \alpha} )^2}{ f_{\theta_{1,n},\psi_n}}(x)dx \\
&\lesssim \|\theta - \theta_{1,n}\|_2^2 + \|(1-\alpha)\psi - \psi_n\|^2 + (1 -  \alpha )^2 \|\psi\|^2\|\psi - \psi_n\|^2 \\
&\lesssim \|\theta - \theta_{1,n}\|_2^2 + \|(1-\alpha)\psi - \psi_n\|^2 + \|\psi_n\| \|\psi\|
\end{split}
\end{equation*}
with a similar inequality for $\int f_n^* (\log f_n^* - \log f_{\theta, \alpha})^2(x) dx $. Hence, with probability bounded by $\epsilon$,  
$$D_n \geq e^{- C_\epsilon} \pi(\tilde S_n) ,$$
for some large positive constant $C_\epsilon$. 
We have the following lower bound on $\pi(\tilde S_n)$. Note that $\|(1 - \alpha)\psi- \psi_n\| + \|\psi_n\| \|\psi\| \leq 1/\sqrt{n}$ implies that $1-\alpha \geq 2 \sqrt{n} \|\psi_n\|^2 $. 
\begin{equation}\label{lb:tildeS}
\begin{split}
\pi\left( \tilde S_n \right) & \gtrsim n^{-d_1/2} n^{-d_\psi/2} \int_{ 2 \sqrt{n} \|\psi_n\|^2}^\delta u^{a_2 - d_\psi-1} du \gtrsim n^{-d_2/2}(\sqrt{n}\|\psi_n\|^2)^{-(d_\psi-a_2)_+} .
\end{split}
\end{equation}

We now bound from above 
 $\pi(S_n(j))$ and  control the entropy of $S_n(j)$ for neighbourhoods $S_{n}(j) = \{ j/\sqrt{n} \leq \| f_{\theta_1, \psi, \alpha} - f_n^* \|_1 \leq  (j+1) /\sqrt{n} \}$.

We have on  $S_n(j)$: 
$$  \| \alpha f_{1,\theta_1}+ (1-\alpha) f_{2, \theta_1, \psi} - f_{2, \theta_{1,n}, \psi_n}\|_1 \leq ( j+1) /\sqrt{n}$$
 We split this set into two subsets $\alpha < 1 -\delta$ for a fixed arbitrarily small $\delta$  and $\alpha \geq 1 -\delta$. In the first case, we have as a first aproximation 
$$\| \alpha f_{1,\theta_1}+ (1-\alpha) f_{2, \theta_1, \psi} - f_{1, \theta_{1,n}}\|_1 \lesssim \|\psi_n\| + j/\sqrt{n}$$ 
which implies in turn that  $\|\theta_1 - \theta_{1,n}\| \lesssim \|\psi_n\| + j/\sqrt{n}$ and $\|\psi\| \lesssim \|\psi_n\|+ j/\sqrt{n}$ where the constants depend on $\delta$.  A more refined Taylor expansion then leads to 
\begin{equation} \label{case2abis}
\begin{split}
& \| \alpha f_{1,\theta_1}+ (1-\alpha) f_{2, \theta_1, \psi} - f_{2, \theta_{1,n}, \psi_n}\|_1  \\
 &= \|(\theta_1 - \theta_{n} )^T (\nabla_\theta f_{1,\theta_{1,n}} + o(1)) + ((1 - \alpha )\psi - \psi_n)^T \nabla_\psi f_{2,\theta_{1,n},0}  + (1-\alpha) (\psi- \psi_n)^T (D_\psi^2 f_{2,\theta_{1,n},0} +o(1))(\psi- \psi_n) \|_1
\end{split}
\end{equation} 
Setting $v_n = \|\theta_1 - \theta_{n} \| + \|(1 - \alpha )\psi - \psi_n\|$ and $\eta = (\theta_1 - \theta_{1,n}, (1-\alpha)\psi-\psi_n)/v_n$, \eqref{case2abis} implies that 
$$\| \alpha f_{1,\theta_1}+ (1-\alpha) f_{2, \theta_1, \psi} - f_{2, \theta_{1,n}, \psi_n}\|_1  \geq v_n\|\eta^T( \nabla f_{2,\theta_{1,n},0} + o(1))  \|_1 + O(\|\psi_n\|^2+v_n^2)$$
and by linear independence of $\nabla f_{2,\theta_1, 0} $ (assumption B4), $v_n \leq C j/\sqrt{n}$ on $S_n(j)$ and 
 $$\pi(S_n(j) \cap \{\alpha \leq 1-\delta\}) \lesssim j^{d_2} n^{-d_2/2}.$$ 
 
We now consider the last case:  $\alpha > 1- \delta $.  As before, a first crude approximation leads to 
$\| f_{1,\theta_1} - f_{2, \theta_{1,n}, \psi_n}\|_1\lesssim 1 -\alpha + j/\sqrt{n}$, which in turns implies that 
$\|\theta_1 - \theta_{1,n}\| + \|\psi_n\| \lesssim 1 - \alpha + j/ \sqrt{n}$. 
In particular if $j/\sqrt{n} = o(\|\psi_n\|)$, then $1-\alpha \gtrsim \|\psi_n\|$. 

Consider first the case where $j/\sqrt{n} \gtrsim \|\psi_n\|$. Then 
\begin{equation}\label{upboundSj1}
\pi(S_n(j) \cap \{1-\alpha \leq \kappa j/\sqrt{n}\}) \lesssim (j/\sqrt{n})^{d_1+a_2}.
\end{equation}
Moreover 
$$\| \alpha f_{1,\theta_1}+ (1-\alpha) f_{2, \theta_1, \psi} - f_{2, \theta_{1,n}, \psi_n}\|_1  \leq  \| \alpha f_{1,\theta}+ (1-\alpha) f_{2, \theta_1, \psi} - f_{1, \theta_{1,n}}\|_1 + O( \|\psi_n\|) $$
so that when $a_2 \leq d_\psi$, \eqref{piSjDn} holds. When $a_2 > d_\psi$, the proof of Theorem \ref{th:cons:embed} implies that there exists $C>0$ such that
$$ S_n(j) \cap \{1-  \alpha > j/\sqrt{n}\} \subset \{ \|\theta_1 - \theta_{1,n}\| + \|(1-\alpha) \psi -\psi_n\| \leq C j/\sqrt{n}\} $$ and
\begin{equation}\label{upboundSj2}
\pi(S_n(j) ) \lesssim (j/\sqrt{n})^{d_2}.
\end{equation}
Now consider the case where $j/\sqrt{n} =o( \|\psi_n\|) $ and recall that $1-\alpha \gtrsim \|\psi_n\|$ and $\|\theta_1 - \theta_{1,n}\| = o(1)$. A Taylor expansion with $\theta_1$ around $\theta_{1,n}$ and $1- \alpha$ around 0  holding $\psi$ fixed  and non null leads to 
\begin{equation*} 
\begin{split}
\| \alpha f_{1,\theta}+ (1-\alpha) f_{2, \theta_1, \psi} - f_{2, \theta_{1,n}, \psi_n}\|_1 &= \|- \psi_n^T\nabla_\psi f_{2, \theta_{1,n}, 0} + (1-\alpha ) (f_{2, \theta_{1,n}, \psi}-f_{2, \theta_{1,n}, \psi_n}) \\
 & \quad +    (\theta_1 - \theta_{1,n})^T\nabla_\theta f_{1, \theta_{1,n}} \|_1 + o(\|\theta_1 - \theta_{1,n}\|+1/n) + O((1-\alpha)\|\psi_n\|)
\end{split}
\end{equation*} 
Set $\eta = (\theta_1 - \theta_{1,n}, 1-\alpha, \psi_n)$ and $\omega = \eta/\|\eta\|$, then 
$$ \|\eta\| \left(  (1-\alpha ) (f_{2, \theta_{1,n}, \psi}-f_{2, \theta_{1,n},0})  +    (\theta_1 - \theta_{1,n})^T\nabla_\theta f_{1, \theta_{1,n}}\right) \|_1 \lesssim j/\sqrt{n}$$ 
so that by linear independence of $(f_{2, \theta_{1,n}, \psi}-f_{2, \theta_{1,n},0}), \nabla_\theta f_{1, \theta_{1,n}}$ (assumption B3), $\|\psi_n\| \lesssim j/\sqrt{n}$ which is impossible. Therefore $\|\psi\| = o(1)$. 
A Taylor expansion then implies that
  \begin{equation}\label{L1ineq}
 \begin{split}
  \|   (\theta_1 - \theta_{1,n})^T \nabla_\theta f_{1,\theta_{1,n}}  &  + ((1-\alpha)\psi- \psi_n)^T \nabla_\psi f_{2, \theta_{1,n}, 0} + (1-\alpha) (\psi-\psi_n)^T D_\psi^2 f_{2,\theta_{1,n}, 0} (\psi-\psi_n) /2\|_1 \\
  & \quad + o(\|\theta_1 - \theta_{1,n}\| +  (1-\alpha) \|\psi\|^2)
    \lesssim  j/\sqrt{n}
   \end{split}
  \end{equation}
When $\| (1-\alpha)\psi- \psi_n\| \leq \delta (1-\alpha) \|\psi-\psi_n\|^2 $, with $\delta$ small,  then 
$$  \|   (\theta_1 - \theta_{1,n})^T \nabla_\theta f_{1,\theta_{1,n}}  + (1-\alpha) (\psi-\psi_n)^T D_\psi^2 f_{2,\theta_{1,n}, 0} (\psi-\psi_n) /2 \|_1\lesssim j/\sqrt{n},$$ 
set $ \eta = (\theta_1 - \theta_{1,n}, \sqrt{1-\alpha} (\psi-\psi_n))$ and $\omega = \eta/\|\eta\|$, assumption B4 implies that 
$$ \|\theta_1 - \theta_{1,n}\| + (1-\alpha) \|\psi-\psi_n\|^2 \lesssim j/\sqrt{n}, \quad  \| (1-\alpha)\psi- \psi_n\| = o(j/\sqrt{n})$$
and since $1-\alpha\leq \delta$ small, $\|\psi-\psi_n\| = \|\psi\| (1 +o(1))$ so that 
$$\|\psi \|\|\psi_n\|\lesssim j/\sqrt{n}, \quad \text{and} \quad \frac{(1-\alpha) j }{ \sqrt{n} \|\psi_n\| } \gtrsim \|(1-\alpha)\psi\| = \|\psi_n\| +o(j/\sqrt{n}),\quad 1-\alpha \gtrsim  \frac{ \sqrt{n} \|\psi_n\|^2  }{j} $$
The prior mass of this set is bounded above by 
$$\pi_{n,2} \lesssim (j/\sqrt{n})^{d_2} \left(\frac{ \sqrt{n} \|\psi_n\|^2  }{j}\right)^{-(d_\psi-a_2)_+}.$$
Similarly when $ (1-\alpha) \|\psi-\psi_n\|^2 \leq \delta\| (1-\alpha)\psi- \psi_n\|   $, \eqref{L1ineq} becomes 
$$  \|   (\theta_1 - \theta_{1,n})^T \nabla_\theta f_{1,\theta_{1,n}}  + ((1-\alpha)\psi- \psi_n)^T \nabla_\psi f_{2, \theta_{1,n}, 0} \|_1\lesssim j/\sqrt{n}, $$ which in turns  implies that 
 $ \| \theta_1 - \theta_{1,n}\| + \| (1-\alpha)\psi- \psi_n\| \lesssim  j/\sqrt{n} $ and $(1-\alpha) \|\psi-\psi_n\|^2= o(j/\sqrt{n} )$ so that $$\frac{\alpha }{ 1-\alpha} \|  \psi_n\| \lesssim \sqrt{ \frac{ j}{\sqrt{n} (1-\alpha)} } $$
 and $1-\alpha \gtrsim (\sqrt{n} \|\psi_n\|^2)/j$ and the prior mass of this set is bounded from above by 
 $$ \pi_{n,3} \lesssim (j/\sqrt{n})^{d_2}  (\sqrt{n} \|\psi_n\|^2)^{- (d_\psi-a_2)_+}j^{(d_\psi-a_2)_+}.$$
 Finally let $\| (1-\alpha)\psi- \psi_n\|  \geq \delta (1-\alpha) \|\psi-\psi_n\|^2 \geq \delta^2\| (1-\alpha)\psi- \psi_n\|  $. Set $\eta_1 = (\theta_1 - \theta_{1,n})/\|\theta_1 - \theta_{1,n}\|$ , $\eta_2 = ((1-\alpha)\psi- \psi_n)/\|(1-\alpha)\psi- \psi_n\|$, $\eta_3 =  (\psi-\psi_n)/\|\psi-\psi_n\|$ and  $u_n = \|\theta_1 - \theta_{1,n}\| + \|(1-\alpha)\psi- \psi_n \| + (1-\alpha) \|\psi-\psi_n\|^2$ and 
 $$w_1 =  \frac{\|\theta_1 - \theta_{1,n}\|}{u_n} , \quad w_2 =  \frac{\|(1-\alpha)\psi- \psi_n\|}{u_n} , \quad  w_3 =  \frac{(1-\alpha) \|\psi-\psi_n\|^2}{u_n},$$ 
 Then $(w_1, w_2, w_3)$ belongs to the sphere in $\mathbf R^3$ with radius 1 and for each $j=1,2,3$, $\eta_j$ belongs to the sphere  (with radius 1 ) in $\mathbf R^d$ with $d= d_\psi$ or $d_1$ so that  \eqref{L1ineq} becomes
\begin{equation*}\label{L2ineq}
 \begin{split}
 u_n  \|   w_1\eta_1^T \nabla_\theta f_{1,\theta_{1,n}}  &  + w_2\eta_2^T \nabla_\psi f_{2, \theta_{1,n}, 0} + w_3\eta_3^T D_\psi^2 f_{2,\theta_{1,n}, 0}\eta_3 \|_1 .
    \lesssim  j/\sqrt{n},
   \end{split}
\end{equation*}
Assumption B4 implies that $u_n \lesssim  j/\sqrt{n}$. This leads to the same constraints as in the case of $\pi_{n,3}$ so that finally 
\begin{equation*}
\pi (S_n(j))  \lesssim (j/\sqrt{n})^{d_2}  + (j/\sqrt{n})^{d_2}  (\sqrt{n} \|\psi_n\|^2)^{- (d_\psi-a_2)_+}j^{(d_\psi-a_2)_+} \quad \mbox{and} \quad \frac{\pi (S_n(j)) }{ \pi(\tilde S_n) }\lesssim j^{d_2+2(d_\psi-a_2)_+}.
 \end{equation*}

 We now control the entropy of $S_n(j)$ for $j \leq M_0 \sqrt{\log n} $, i.e. the logarithm of the covering number of $S_n(j)$ by $L_1$ balls with radius $\zeta j /\sqrt{n}$, $\zeta >0$ arbitrarily small. 
Recall that from the above conditions $S_n(j)$ is included in 
 $$\{\| \theta - \theta_{1,n}\| \leq C_1 j/\sqrt{n} \}\cap\{\| (1- \alpha)\psi -\psi_n\|\leq C_1j/\sqrt{n} \}\cap \{ (1-\alpha) \| \psi - \psi_n\|^2 \leq C_1 j/\sqrt{n}\}$$
 for some $C_1>0$ large enough, so that if $\|\psi_n\| \lesssim j/\sqrt{n}$, we are back to the proof of Theorem \ref{th:cons:embed}, and the local entropy is bounded by a constant. If $j/\sqrt{n} \leq \delta \|\psi_n\|$ with $\delta $ small, recall from the above computations that $1-\alpha \geq \tau \sqrt{n}\|\psi_n\|^2/j$ for some $\tau >0$ and $\|\psi - \psi_n\| \leq C_1 j (1-\alpha)^{-1}/\sqrt{n}$. A Taylor expansion in $\psi, \psi', \theta_1, \theta_1'$ leads to 
 \begin{equation*}
 \begin{split}
&  \| \alpha f_{1,\theta} +(1-\alpha)f_{2, \theta, \psi}- ( \alpha' f_{1,\theta'} +(1-\alpha')f_{2, \theta', \psi'}) \|_1\\
&=\|( (1-\alpha)\psi -(1-\alpha')\psi')^T \nabla_\psi f_{2, \theta_1, 0} + \frac{(1-\alpha)}{2}\psi^T D^2_\psi f_{2, \theta_1, \bar \psi} \psi - \frac{(1-\alpha')}{2}(\psi')^T D^2_\psi f_{2, \theta_1, \bar \psi'} \psi' \|+O( \|\theta_1 - \theta_1'\| ) + o(j/\sqrt{n})\\
&=\|( (1-\alpha)\psi -(1-\alpha')\psi')^T \nabla_\psi f_{2, \theta_1, 0} + \frac{1}{2}[(1-\alpha)\psi - (1-\alpha')\psi']^T D^2_\psi f_{2, \theta_1, \bar \psi} \psi + \frac{(1-\alpha')}{2}(\psi')^T D^2_\psi f_{2, \theta_1, \bar \psi'} (\psi - \psi') \|_1 \\
& \quad + O( \|\theta_1 - \theta_1'\|) +  o(j/\sqrt{n}) \leq \frac{\zeta j }{\sqrt{n}}
 \end{split}
 \end{equation*}
as soon as 
$\|\theta-\theta'\|\leq \delta  \zeta j /(4\sqrt{n})$, $\|(1-\alpha)\psi -(1-\alpha')\psi'\|\leq \delta \zeta j /(4\sqrt{n})$ and $\|\psi'-\psi\| \leq \delta \zeta j /(4\sqrt{n} \|\psi_n\|)$. The number of sets to cover 
$$\{ \|(1-\alpha)\psi - \psi_n\| \leq C_1j/\sqrt{n}\} \cap \{(1-\alpha) \|\psi - \psi_n\|^2 \leq C_1j/\sqrt{n}\}\cap \{\|\theta_1 - \theta_{1,n} \| \leq C_1j/\sqrt{n}\}$$
is bounded by a constant independent of $j$ and $n$. 
Finally combining  the lower bound on $D_n$, the upper bound on $\pi(S_n(j))$, the entropy bounds above and Theorem 2.4 of \citep{ghosal:ghosh:vdv:00}, we obtain that for all increasing sequence to infinity $z_n$,  
uniformly over  $\|\psi_n\| \geq M_n/\sqrt{n}$, $\theta_{1,n}  \in \Theta_1$
$$E_{f_n^*}\left[ \pi\left( \| f_{\alpha, \theta, \psi} - f_n^*\|_1\geq z_n /\sqrt{n} |\mathbf x^n\right)\right] = o(1), \quad f_n^* = f_{2, \theta_{1,n}, \psi_n}. $$

From the  above computations with $j = z_n $ going to infinity with $z_n  = o(\sqrt{n}\|\psi_n\|) = o(n^{1/4}) $ we have that $1-\alpha \geq C_2 \sqrt{n} \| \psi_n\|^2$, so that there exists   $M_0>0$ for which 
$$\sup_{\|\psi\| \geq M_n/\sqrt{n}} \sup_{\theta_1 \in \Theta_1} E_{\theta_{1}, \psi} \left[ \pi(1-\alpha < M_0M_n^2/\sqrt{n} | \mathbf x^n)\right] = o(1)$$
and Theorem \ref{th:separationrate} is proved. 

Note that when $a_2 >d_\psi$, for any $e_n>0$ going to 0,  the posterior probability of the set $A_n = \{1-\alpha \leq e_n\}$ has posterior probability going to 0 under $f_n^*$ since the prior mass of the event 
$$ \{\|\theta_1 - \theta_{1,n}\|+\|(1-\alpha)\psi - \psi_n\|+ (1-\alpha)\|\psi - \psi_n\|^2 \leq z_n/\sqrt{n}\} \cap \{ (1-\alpha)\leq e_n\}$$
is of order $O(e_n^{a_2-d_\psi}z_n^{d_2}n^{-d_2/2}) = o(n^{-d_2/2})$ as soon as $e_n^{a_2-d_\psi}z_n^{d_2}=o(1)$. Since we can choose $z_n$ going to infinity arbitrarily slowly, the result holds. 
\end{document}